\documentclass[english,twoside,11pt]{article}
\pdfoutput=1

\usepackage[lmargin=1in,rmargin=1in,tmargin=1in,bmargin=1in]{geometry}
\usepackage{float}
\usepackage{graphicx}
\usepackage{color}
\usepackage{tikz}
\usepackage[margin=20pt]{caption}

\usepackage{amsmath}
\usepackage{amssymb}
\usepackage{amsthm}
\usepackage{mathrsfs}

\usepackage[english]{babel}
\usepackage{flafter}
\usepackage{array}
\usepackage{paralist}

\graphicspath{{.}{graphics/}}

\newtheorem{theorem}{Theorem}

\newtheorem*{corollary*}{Corollary}
\newtheorem{lemma}[theorem]{Lemma}
\newtheorem{proposition}[theorem]{Proposition}
\newtheorem*{proposition*}{Proposition}

\newtheorem*{property*}{Property}

\theoremstyle{definition}
\newtheorem{definition}[theorem]{Definition}
\newtheorem*{definition*}{Definition}

\theoremstyle{remark}

\newtheorem*{remark*}{Remark}

\newtheorem*{example*}{Example}

\definecolor{darkblue}{rgb}{0,0,0.38}
\definecolor{darkred}{rgb}{0.6,0,0}
\definecolor{darkgreen}{rgb}{0.1,0.35,0}

\def\ALG{\mathsf{ALG}}

\newcommand{\opt}[1]{\mathrm{OPT}(#1)}
\newcommand{\cov}[1]{\mathrm{cov}(#1)}
\newcommand{\nca}[1]{\mathrm{nca}(#1)}
\newcommand{\supp}[1]{\mathrm{supp}(#1)}

\title{\bf Improved Approximation for Weighted Tree Augmentation with Bounded Costs}

\author{
    David Adjiashvili \\
    \small Institute for Operations Research, ETH Z\"urich \\
    \small R\"amistrasse 101, 8092 Z\"urich, Switzerland \\
    \small email: addavid@ethz.ch
}

\date{}
 
\begin{document}

\maketitle

\thispagestyle{empty}

\abstract{
The Weighted Tree Augmentation Problem (WTAP) is a fundamental well-studied problem in the field of network design. Given an undirected tree $G=(V,E)$, an additional set of edges $L \subseteq V\times V$ disjoint from $E$ called \textit{links}, and a cost vector $c\in \mathbb{R}_{\geq 0}^L$, WTAP asks to find a minimum-cost set $F\subseteq L$ with the property that $(V,E\cup F)$ is $2$-edge connected. The special case where $c_\ell = 1$ for all $\ell\in L$ is called the Tree Augmentation Problem (TAP). Both problems are known to be NP-hard.

For the class of bounded cost vectors, we present a first improved approximation algorithm for WTAP since more than three decades. Concretely, for any $M\in \mathbb{Z}_{\geq 1}$ and $\epsilon > 0$, we present an LP based $(\delta+\epsilon)$-approximation for WTAP restricted to cost vectors $c$ in $[1,M]^L$ for $\delta \approx 1.96417$. For the special case of TAP we improve this factor to $\frac{5}{3}+\epsilon$.

% Although they are fundamental in the field of network design, the approximability of 
% both TAP and WTAP is still not well understood. The best known approximation 
% algorithms for TAP and WTAP achieve factors $\frac{3}{2}$ and $2$, respectively,
% both not known to be tight. Achieving a factor better $2$ for WTAP is one 
% of the central open problem in the field, a task that seemed out of reach 
% for the known techniques achieving an factor better than $2$ for TAP, which
% are combinatorial in nature.

% In an effort to improve the factor $2$ for WTAP several recent works
% analyzed LP- and SDP-relaxations of TAP, achieving improved factors,
% including an LP-based $\frac{7}{4}$ approximation~\cite{1.75}, and an SDP-based 
% $\frac{3}{2}+\epsilon$ approximation~\cite{sdp}. We continue this line of work
% by presenting for any two constants $M\in \mathbb{Z}_{\geq 1}$ and $\epsilon > 0$
% % \begin{itemize}
%  (i) An LP-based $\delta+\epsilon$ approximation for WTAP with costs in the range
% $\{1,2,\cdots,M\}$ for $\delta \approx 1.96417$, and
%  (ii) An LP-based $\frac{5}{3}+\epsilon$ approximation for TAP.
% % \end{itemize}

Our results rely on a new LP, that significantly differs from existing LPs achieving improved bounds for TAP. We round a fractional solution in two phases. The first phase uses the fractional solution to decompose the tree and its fractional solution into so-called $\beta$-simple pairs losing only an $\epsilon$-factor in the objective function. We then show how to use the additional constraints in our LP combined with the $\beta$-simple structure to round a fractional solution in each part of the decomposition. 
% To the best of our knowledge, our algorithm for
% WTAP is the first approximation algorithm with factor better than $2$ for any
% non-trivial class of cost functions and general trees.
}

\newpage

\setcounter{page}{1}

\section{Introduction}

The Weighted Tree Augmentation Problem (WTAP) is a well-studied 
problem in the field of network design. Given an 
undirected tree $G=(V,E)$, an additional set 
of edges $L \subseteq V\times V$ disjoint from $E$ 
called \textit{links}, and a cost vector $c\in \mathbb{R}_{\geq 0}^L$,
WTAP asks to find a minimum-cost set $F\subseteq L$ with 
the property that $(V,E\cup F)$ is $2$-edge connected. Recall
that a graph is $2$-edge connected if there are at least
$2$ edge-disjoint paths between any two nodes. The special 
case where $c_\ell = 1$ for all $\ell\in L$ is called the
Tree Augmentation Problem (TAP). 

WTAP is widely recognized as one of the fundamental problems in
the field of network design (see e.g. the surveys of
Kuhler~\cite{KuhlerSurvey} and Kortsarz and Nutov~\cite{kortsarz2010approximating}).
The main open problem is whether there is an approximation algorithm 
for WTAP with factor better than $2$ (we review the literature in the next section).
Our main result is an improved approximation algorithm for the 
case of bounded costs.

\begin{theorem}\label{thm:wtap}
Let $\delta = \frac{8(23+3\sqrt{5})}{121} \approx 1.96418$. For any 
fixed $M\in \mathbb{Z}_{\geq 1}$ and $\epsilon \in \mathbb{R}_{>0}$, 
there exists an LP-based polynomial-time $(\delta+\epsilon)$-approximation algorithm for WTAP restricted
to instances with cost vectors satisfying $c\in [1,M]^L$, i.e.\
whenever $1\leq c_\ell\leq M$ holds for all $\ell\in L$.
\end{theorem}
  
%The approximation factor in Theorem~\ref{thm:wtap}
%is obtained by balancing the approximation factors obtained from the rounding procedures 
%used for the two types of trees in the decomposition. 
%For TAP we are able to significantly improve the guarantee of rounding 
%one type of trees, resulting in the following theorem.
%(Appendix~\ref{apx:tap}).
For TAP we obtain an improved approximation guarantee as stated hereafter.

\begin{theorem}\label{thm:tap}
 For any fixed $\epsilon \in \mathbb{R}_{>0}$, 
there exists an LP-based polynomial $\left(\frac{5}{3}+\epsilon\right)$-approximation algorithm for TAP.
\end{theorem}

Both algorithms achieve approximations with respect to a new linear
programming (LP) relaxation, which we call \emph{the bundle LP}.
To the best of our knowledge, our result for WTAP is the first approximation 
algorithm that achieves a factor better than $2$ for 
all trees, and any non-trivial family of cost functions.
Furthermore, both bounds are the best known among those that are based on 
an LP. 

Our results are based on several new ideas combined with classical 
results for WTAP, which we briefly explain next.
First, we explain the well-known set covering reformulation of WTAP.
We denote by $V[H]$ and $E[H]$ the node set and the edge set of a graph $H$, respectively.
Associate with every link $\ell = uv \in L$ the unique path $P_{uv} \subseteq E[G]$ 
in $G$ connecting $u$ and $v$. Now, it is easy to verify that a set of links $S\subseteq L$
is a feasible solution for the WTAP instance at hand, if and only if the union
of the corresponding paths covers the edge set of $G$, namely if
$\cup_{\ell\in S} P_\ell = E[G]$. 

For a set $X\subseteq E[G]$ denote by $\cov{X} \subseteq L$ the set 
of links that cover at least one edge of $X$. When $X = \{e\}$ is a
singleton we write $\cov{e}$ instead of $\cov{\{e\}}$. 
The \emph{natural LP relaxation} for WTAP, also known
as \emph{the cut LP}, is based on the latter reformulation.
It contains one variable $x_\ell$ 
for each link $\ell \in L$ and asks to solve
\begin{eqnarray}
\text{\textrm{minimize}} \,\,\, \sum_{\ell\in L} c_\ell x_\ell & & \text{subject to}\\
\label{eq:constraint1}
\sum_{\ell \in \cov{e}} x_\ell &\geq& 1 \quad\quad \forall e\in E[G],\\
\label{eq:constraint2}
x_\ell &\geq& 0  \quad\quad  \forall \ell\in L.
\end{eqnarray}

The bundle LP, the optimal solution of which we round to obtain our results,
is defined as follows. 
For an integer $\gamma \in \mathbb{Z}_{\geq 1}$, a
\textit{$\gamma$-bundle} is a union of $\gamma$ (not necessarily 
distinct) paths in $G$. See Figure~\ref{fig:bundle} for an example. 
Denote by $\mathcal{B}_{\gamma}$ the set of all $\gamma$-bundles in $G$. 
In words, for a carefully chosen $\gamma$,
the bundle LP contains, on top of the constraints from the natural LP,
constraints that ensure that each $\gamma$-bundle
is covered in the fractional solution by links with sufficiently
high cost. Formally, the LP contains the constraints
\begin{equation}\label{eq:bundle-constraints}
 \sum_{\ell\in \cov{B}} c_\ell x_\ell \geq \opt{B} \quad \quad \forall B\in \mathcal{B}_\gamma,
\end{equation}
where for any $X\subseteq E[G]$, $\opt{X} \in \mathbb{R}_{\geq 0}$ is the minimum cost 
of a set of links in $L$ that covers all edges in $X$.
The latter constraints are clearly valid for any integral solution, as any such
solution contains a feasible solution covering the edges in $B$.
The $\gamma$-bundle LP is then given by
$$
\text{\textrm{minimize}} \,\,\, \sum_{\ell\in L} c_\ell x_\ell \quad \text{subject to} \quad (\ref{eq:constraint1}), (\ref{eq:constraint2}) \,\,  
\text{and} \,\, (\ref{eq:bundle-constraints}). \quad \quad \quad \quad (\mathrm{LP}_\gamma)
$$

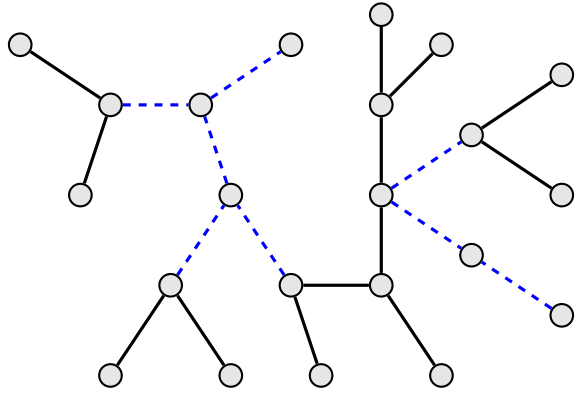
\begin{figure}[h]
\begin{center}
\begin{tikzpicture}[scale=0.4]

\node (a1) at (3,0) [circle,draw=black!100,fill=black!10,thick,inner sep=1pt,minimum size=3mm] {};
\node (a2) at (5,3) [circle,draw=black!100,fill=black!10,thick,inner sep=1pt,minimum size=3mm] {};
\node (a3) at (2,6) [circle,draw=black!100,fill=black!10,thick,inner sep=1pt,minimum size=3mm] {};
\node (a4) at (3,9) [circle,draw=black!100,fill=black!10,thick,inner sep=1pt,minimum size=3mm] {};
\node (a5) at (0,11) [circle,draw=black!100,fill=black!10,thick,inner sep=1pt,minimum size=3mm] {};
\node (a6) at (7,0) [circle,draw=black!100,fill=black!10,thick,inner sep=1pt,minimum size=3mm] {};
\node (a7) at (9,3) [circle,draw=black!100,fill=black!10,thick,inner sep=1pt,minimum size=3mm] {};
\node (a8) at (7,6) [circle,draw=black!100,fill=black!10,thick,inner sep=1pt,minimum size=3mm] {};
\node (a9) at (6,9) [circle,draw=black!100,fill=black!10,thick,inner sep=1pt,minimum size=3mm] {};
\node (a10) at (9,11) [circle,draw=black!100,fill=black!10,thick,inner sep=1pt,minimum size=3mm] {};
\node (a11) at (10,0) [circle,draw=black!100,fill=black!10,thick,inner sep=1pt,minimum size=3mm] {};
\node (a12) at (12,3) [circle,draw=black!100,fill=black!10,thick,inner sep=1pt,minimum size=3mm] {};
\node (a13) at (14,0) [circle,draw=black!100,fill=black!10,thick,inner sep=1pt,minimum size=3mm] {};
\node (a14) at (12,6) [circle,draw=black!100,fill=black!10,thick,inner sep=1pt,minimum size=3mm] {};
\node (a15) at (12,9) [circle,draw=black!100,fill=black!10,thick,inner sep=1pt,minimum size=3mm] {};
\node (a16) at (12,12) [circle,draw=black!100,fill=black!10,thick,inner sep=1pt,minimum size=3mm] {};
\node (a17) at (18,2) [circle,draw=black!100,fill=black!10,thick,inner sep=1pt,minimum size=3mm] {};
\node (a18) at (15,4) [circle,draw=black!100,fill=black!10,thick,inner sep=1pt,minimum size=3mm] {};
\node (a19) at (18,6) [circle,draw=black!100,fill=black!10,thick,inner sep=1pt,minimum size=3mm] {};
\node (a20) at (15,8) [circle,draw=black!100,fill=black!10,thick,inner sep=1pt,minimum size=3mm] {};
\node (a21) at (18,10) [circle,draw=black!100,fill=black!10,thick,inner sep=1pt,minimum size=3mm] {};
\node (a22) at (14,11) [circle,draw=black!100,fill=black!10,thick,inner sep=1pt,minimum size=3mm] {};

\draw [-,black,very thick] (a1) to node [black] {} (a2);
\draw [-,black,very thick] (a2) to node [black] {} (a6);
\draw [-,blue,very thick,dashed] (a2) to node [black] {} (a8);
\draw [-,black,very thick] (a3) to node [black] {} (a4);
\draw [-,black,very thick] (a4) to node [black] {} (a5);
\draw [-,blue,very thick,dashed] (a4) to node [black] {} (a9);
\draw [-,blue,very thick,dashed] (a7) to node [black] {} (a8);
\draw [-,blue,very thick,dashed] (a8) to node [black] {} (a9);
\draw [-,blue,very thick,dashed] (a9) to node [black] {} (a10);
\draw [-,black,very thick] (a7) to node [black] {} (a11);
\draw [-,black,very thick] (a7) to node [black] {} (a12);
\draw [-,black,very thick] (a12) to node [black] {} (a13);
\draw [-,black,very thick] (a12) to node [black] {} (a14);
\draw [-,black,very thick] (a14) to node [black] {} (a15);
\draw [-,black,very thick] (a15) to node [black] {} (a16);
\draw [-,black,very thick] (a15) to node [black] {} (a22);
\draw [-,blue,very thick,dashed] (a14) to node [black] {} (a18);
\draw [-,blue,very thick,dashed] (a14) to node [black] {} (a20);
\draw [-,blue,very thick,dashed] (a17) to node [black] {} (a18);
\draw [-,black,very thick] (a19) to node [black] {} (a20);
\draw [-,black,very thick] (a20) to node [black] {} (a21);

\end{tikzpicture}
\end{center}
\caption{The dashed green edges can be obtined as a union of three paths, hence they
comprise a $3$-bundle.}\label{fig:bundle}
\end{figure}

% It is easy to see that the bundle constraints for $\gamma \geq 1$ imply the
% cover constraints $(\ref{eq:constraint1})$, but we state $\mathrm{LP}_{\gamma}$
% with both sets of constraints for simplicity.
Solving the bundle LP entails calculating the values $\opt{B}$ for all $B\in \mathcal{B}_\gamma$.
We show how this can be done in polynomial time whenever $\gamma$ is constant, and
in time $n^{\gamma^{O(1)}}$ in general, in Appendix~\ref{apx:solve-lp}.

Intuitively, the bundle LP cuts off all fractional solutions that have large 
integrality gap due to "simple obstructions", including, but not
restricted to, subtrees with few leaves. A large integrality gap can 
hence result only due to more "global substructures" of the tree.
To exploit this feature, our strategy is to decompose the instance at
hand and its fractional solution into parts that are so small, that they can no longer 
contain such global structures, and hence they cannot suffer from large integrality gap.
We achieve this goal using a simple two-step decomposition, which first transforms the
solution pair into a ``thin solution,'' namely one that does not over-cover any edge, and
then greedily breaks the tree into simpler trees. We manage to lose only an 
$\epsilon$-fraction in terms of cost in this decomposition.

The obtained smaller trees, equipped with corresponding fractional solutions, are then
proved to be of one of two types: Either they are already ``sufficiently small'' to
apply the bundle constraints, or they are sufficiently close to instances that 
have another special structure, which we call \emph{star-shaped} instances. Intuitively,
star-shaped instances are WTAP instances that can be modeled as \emph{edge-cover problems},
implying that a fractional solution can be rounded with a loss of a factor significantly
better than $2$.

We remark that, although our algorithm for WTAP is only polynomial 
when $M = \max_{\ell\in L} c_\ell$ is constant, it can also be used to obtain the 
same approximation guarantee for non-constant $M$, with a running time of $n^{M^{O(1)}}$,
where $n = |V[G]|$.

\subsection{Related Work}\label{sec:relwork}

Frederickson and J\'{a}J\'{a}~\cite{frederickson1981approximation}
proved that WTAP is NP-hard, even for trees with constant 
diameter. The restricted special case of TAP, where the links
form a cycle on the leaves of the tree is also NP-hard, as was
shown by Cheriyan et al.~\cite{cheriyan19992}.

The best known approximation for WTAP is an elegant $2$-approximation due to
Frederickson and J\'{a}J\'{a}~\cite{frederickson1981approximation}, which
was later further simplified by Khuller and Thurimella~\cite{khuller1993approximation}.
Numerous algorithms have since been developed achieving the same factor.
These include, among others, the iterative rounding algorithm of Jain~\cite{Jain} and
the primal-dual algorithm of Goemans et al.~\cite{goemans1994improved}.
While these algorithms are designed for much more general network 
design problems, the factor $2$ is tight for them, even in 
the case of TAP. The factor $2$ has since only been improved for special
classes of trees, including a $(1+\ln 2)$-approximation for the case of 
bounded-diameter trees by Cohen and Nutov~\cite{cohen2013}.
Improving the factor $2$ for WTAP is a major open problem in network 
design~\cite{KuhlerSurvey,kortsarz2010approximating}.
To the best of our knowledge, our algorithm for WTAP is the first 
improvement in the approximation guarantee (over the factor $2$) 
for any special case of WTAP and all trees for over three decades.

In contrast, several approximation algorithms with factors better than
$2$ are known for TAP. The first such algorithm, achieving a factor
$1.875 + \epsilon$ was given by Nagamochi~\cite{nagamochi1dot875}.
This was later improved to $1.8$ by Even et al.~\cite{even20091}. The
current best algorithm is a $\frac{3}{2}$-approximation, due to 
Kortsarz and Nutov~\cite{kortsarz2016simplified}. An improved
$\frac{17}{12}$-approximation was developed by Maduel and 
Nutov~\cite{maduel2010covering} for the special
case of TAP, where each link connects two leaves of the tree.

The techniques used to achieve the latter improved algorithms for 
TAP are combinatorial in nature, and seem very hard to modify for 
an improved approximation of WTAP.
In an effort to improve the approximation factor for WTAP, several
algorithm with approximation factors better than $2$ have recently 
been developed for TAP, that are based on continuous relaxations of the problem. 
Along these lines, 
Kortsarz and Nutov~\cite{KortsarzNutov2016} recently showed an LP-based
$\frac{7}{4}$-approximation algorithm. Our 
$\left(\frac{5}{3}+\epsilon\right)$-approximation for TAP improves 
on that factor for LP-based algorithms.
Cheriyan and Gao~\cite{CheriyanG15} presented an approximation with respect
to a semidefinite program (SDP) obtained from Lasserre tightening of an 
LP relaxation with factor $\frac{3}{2}+\epsilon$. A strong
point of both papers compared to our result is that the mathematical program 
is used in the analysis, while the algorithms themselves are combinatorial. 

The bundle LP differs from all existing LPs that were recently proved to 
have integrality gap better than $2$.
%, in the case of TAP, have integrality gap that is  provably better than $2$. 
Indeed, most existing such LPs include, on top of the constraints
of the natural LP, constraints that exploit structural properties of feasible, or optimal 
solutions restricted to very specific structures, such as \emph{twin links, stems} etc. 
(see~\cite{KortsarzNutov2016} for formal definitions of these and other related notions).
%These constraints cut off some undesirable fractional solutions.
In contrast, as we discussed before, the bundle LP attempts to uniformly cut 
off fractional solutions that have low cost due to \emph{all} sufficiently 
simple obstructions which have large integrality gap. 
%As a result, if the fractional solution of the bundle LP has bad integrality gap, 
%this is a result of a ``global structure'' and not of a union of small obstructions. 

Another important open problem is the integrality gap of the
natural LP relaxation.
%derived from the set covering interpretation of WTAP, which we explain next.
%We denote by $V[H]$ and $E[H]$ the node set and the edge set of a graph $H$, respectively.
%Associate with every link $\ell = uv \in L$ the unique path $P_{uv} \subseteq E[G]$ 
%in $G$ connecting $u$ and $v$. Now, it is easy to verify that a set of links $S\subseteq L$
%is a feasible solution for the WTAP instance at hand, if and only if the union
%of the corresponding paths covers the edge set of $G$, namely if
%$\cup_{\ell\in S} P_\ell = E[G]$. 
% The latter interpretation allows as
% to use set covering terminology, and say that certain links ``cover''
% certain edges etc. 
% This interpretation also motivates a 
%The natural LP relaxation, known as \emph{the cut LP}, is simply the fractional relaxation
%of the integer programming formulation of the latter set cover problem. 
An upper bound
of $2$ on the integrality gap of the cut LP follows from some of the
aforementioned $2$-approximations, including that of Jain~\cite{Jain} and
Goemans et al.~\cite{goemans1994improved}. There is also a lower bound
of $\frac{3}{2}$ known on the integrality gap due to Cheriyan et al.~\cite{IntegralityGap}. 
The true integrality gap, however, still remains open, even for the special
case of TAP.

WTAP can also be interpreted as a problem of covering a laminar family with
point-to-point links (see e.g.~\cite{KortsarzNutov2016}). This observation implies
that the problem of augmenting a $k$-edge connected graph to a $(k+1)$-edge
connected graph can be reduced to WTAP, whenever $k$ is odd, due to the 
laminar structure of the family of minimum cuts, in this case.

\section{Preliminaries}\label{sec:prelims}

We delay certain proofs of technical results
to Appendix~\ref{apx:proofs}, and the treatment of the improved rounding
algorithm for TAP to Appendix~\ref{apx:tap}.

Our algorithm relies on a few basic results for WTAP and
some related problems that we recap here. For a finite set $N$ and 
a vector $x\in \mathbb{R}^N_{\geq 0}$, let $\supp{x} = \{i\in N \,\mid\, x_i > 0\}$
denote the support of $x$. For a subset $Y\subseteq N$
denote $x(Y) = \sum_{i\in Y} x_i$. For an integer $k\in \mathbb{Z}_{\geq 1}$
denote $[k] = \{1,\cdots, k\}$. To distinguish between edges of the
tree and links, we write $e =\{u,v\}$ and $\ell=uv$ for an edge $e$ connecting $u$
and $v$ and a link $\ell$ connecting $u$ and $v$, respectively.

Throughout the paper we use the notion of \emph{contraction} of edges
and links. By contracting an edge $e = \{u,v\} \in E[G]$ we mean replacing
the nodes $u$ and $v$ by a new \emph{compound node} $w$ and connecting every
edge $e'\in E[G]$ different from $e$ that had a connection to either 
$u$, or $v$ to the new node $w$. We also update the set of links in an
analogous way: A link that was connected to either $u$, or $v$ is connected
after contracting $e$ to $w$. If the link was connected to both $u$ and 
$v$ it becomes a self-loop of $w$. By contracting a set of edges $F\subseteq E$ 
we mean contracting the edges in $F$ one after the other in any order.
By contacting a link (or a set of links) we mean contracting the set of
edges of the tree covered by the link (or set of links). Figure~\ref{fig:contraction}
illustrates the contraction operation.

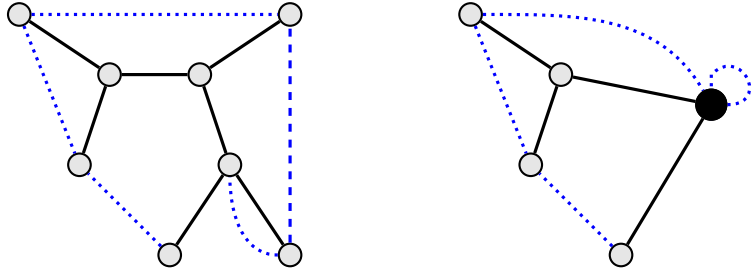
\begin{figure}[h]
\begin{center}
\begin{tikzpicture}[scale=0.4]

\begin{scope}
\node (a2) at (5,3) [circle,draw=black!100,fill=black!10,thick,inner sep=1pt,minimum size=3mm] {};
\node (a3) at (2,6) [circle,draw=black!100,fill=black!10,thick,inner sep=1pt,minimum size=3mm] {};
\node (a4) at (3,9) [circle,draw=black!100,fill=black!10,thick,inner sep=1pt,minimum size=3mm] {};
\node (a5) at (0,11) [circle,draw=black!100,fill=black!10,thick,inner sep=1pt,minimum size=3mm] {};
\node (a7) at (9,3) [circle,draw=black!100,fill=black!10,thick,inner sep=1pt,minimum size=3mm] {};
\node (a8) at (7,6) [circle,draw=black!100,fill=black!10,thick,inner sep=1pt,minimum size=3mm] {};
\node (a9) at (6,9) [circle,draw=black!100,fill=black!10,thick,inner sep=1pt,minimum size=3mm] {};
\node (a10) at (9,11) [circle,draw=black!100,fill=black!10,thick,inner sep=1pt,minimum size=3mm] {};

%edges
\draw [-,black,very thick] (a2) to node [black] {} (a8);
\draw [-,black,very thick] (a3) to node [black] {} (a4);
\draw [-,black,very thick] (a4) to node [black] {} (a5);
\draw [-,black,very thick] (a4) to node [black] {} (a9);
\draw [-,black,very thick] (a7) to node [black] {} (a8);
\draw [-,black,very thick] (a8) to node [black] {} (a9);
\draw [-,black,very thick] (a9) to node [black] {} (a10);

%links

\draw [-,blue,very thick,dotted] (a2) to node [black] {} (a3);
\draw [-,blue,very thick,dotted] (a5) to node [black] {} (a10);
\draw [-,blue,very thick,dotted,out=180,in=-90] (a7) to node [black] {} (a8);
\draw [-,blue,very thick,dashed] (a7) to node [black] {} (a10);
\draw [-,blue,very thick,dotted] (a5) to node [black] {} (a3);
\end{scope}

%after contraction
\begin{scope}[xshift=15cm]
\node (a2) at (5,3) [circle,draw=black!100,fill=black!10,thick,inner sep=1pt,minimum size=3mm] {};
\node (a3) at (2,6) [circle,draw=black!100,fill=black!10,thick,inner sep=1pt,minimum size=3mm] {};
\node (a4) at (3,9) [circle,draw=black!100,fill=black!10,thick,inner sep=1pt,minimum size=3mm] {};
\node (a5) at (0,11) [circle,draw=black!100,fill=black!10,thick,inner sep=1pt,minimum size=3mm] {};
% \node (a7) at (9,3) [circle,draw=black!100,fill=black!10,thick,inner sep=1pt,minimum size=3mm] {};
% \node (a8) at (7,6) [circle,draw=black!100,fill=black!10,thick,inner sep=1pt,minimum size=3mm] {};
% \node (a9) at (6,9) [circle,draw=black!100,fill=black!10,thick,inner sep=1pt,minimum size=3mm] {};
\node (b) at (8,8) [circle,draw=black!100,fill=black!100,thick,inner sep=1pt,minimum size=4mm] {};

%edges
\draw [-,black,very thick] (a2) to node [black] {} (b);
\draw [-,black,very thick] (a3) to node [black] {} (a4);
\draw [-,black,very thick] (a4) to node [black] {} (a5);
\draw [-,black,very thick] (a4) to node [black] {} (b);

%links

% \Loop[style=dotted,color=blue,dist = 4cm, dir = NO,thickness=thick](b.west)
% \draw (b) to [out=330,in=300,looseness=8] (b)
% \draw [-,blue,very thick,dotted] (b) .. controls (10,8) and (10,10) .. controls (8,10) and (8,10) .. (b);
% \draw [-,blue,very thick,dotted] (b) to node [black,loopness=8] {} (b);
\draw [-,blue,very thick,dotted,out=0,in=90] (b) .. controls +(2,0)  and +(1,1) .. +(0,1) .. controls +(0,-0.5) .. (b);
\draw [-,blue,very thick,dotted,out=0,in=120] (a5) to node [black] {} (b);
\draw [-,blue,very thick,dotted] (a2) to node [black] {} (a3);
\draw [-,blue,very thick,dotted] (a5) to node [black] {} (a3);

\node (bdummy) at (8,8) [circle,draw=black!100,fill=black!100,thick,inner sep=1pt,minimum size=4mm] {};

\end{scope}

\end{tikzpicture}
\end{center}
\caption{An illustration of a contraction operation. The links are shown with dotted and 
dashed lines. The dased link in the left instance is contracted to obtain the instance 
on the right. The full black node is the obtained compound node}\label{fig:contraction}
\end{figure}

We refer to the book of Schrijver~\cite{schrijver2002combinatorial} for further details and
results on some of the notions presented in this section. Furthermore, throughout the 
paper we do not optimize the running time of algorithms to facilitate a cleaner presentation.
Next, we introduce the well-known notions of shadows of links and shadow completeness.

\paragraph{Shadows and shadow completion.}
For a tree $G$ and a link $\ell = uv \in V[G] \times V[G]$ denote 
by $P^G_\ell \subseteq E[G]$ the $u$-$v$ path in $G$. When $G$ is clear 
from the context we drop the superscript and write $P_\ell$.
Let $\ell, \ell'\in V[G] \times V[G]$ be two links. We say that 
$\ell'$ is a \emph{shadow} of $\ell$ if $P_{\ell'}\subseteq P_\ell$.
An instance $(G,L)$ is \emph{shadow complete} if for every $\ell\in L$
all shadows of $\ell$ are also in $L$. We can always assume that the
instance is shadow complete: If it is not, then all shadows of links
in $L$ can be added to $L$. A cost of an added link $\ell$ is the minimum
cost of any original link in $L$, of which $\ell$ is a shadow. The
latter shadow completing operation does not change the optimal value
of the instance. Furthermore, shadows that are added in the completion
can always be replaced in any solution by original links that cover 
a superset of the edges covered by the shadow without increasing the 
cost of the solution. We will hence always assume that the instance
is shadow complete.

\paragraph{Up-links and a simple LP-based $2$-approximation.}
One important building block that is used throughout the algorithm is a 
simple $2$-approximation algorithm for WTAP that rounds a fractional 
solution $x\in \mathbb{R}_{\geq 0}^L$ for the natural LP relaxation. 
As we mentioned before, the latter is achievable with several existing
algorithms. We present a simple such algorithm here for completeness 
and since it outlines ideas that we use later on.

The algorithm starts by rooting the tree $G$ at an arbitrary node $r\in V[G]$.
For every link $\ell = uv \in L$ denote by $\nca{uv}\in V[G]$ the nearest 
common ancestor of $u$ and $v$ in $G$. We call a link $\ell = uv$ an 
\emph{up-link} if $\nca{uv} \in \{u,v\}$. Denote by $L^{up} \subseteq L$
the set of all up-links. The following simple rounding lemma is driving
the approximation algorithm.

\begin{lemma}\label{lem:up_rounding}
 Let $x\in \mathbb{R}_{\geq 0}^L$ be a feasible solution of the natural LP
on instance $(G,L,c)$ with the property $\supp{x}\subseteq L^{up}$. Then 
there is feasible set of links $S\subseteq L^{up}$ with cost 
$c(S) \leq c^\intercal x$, that can be computed in polynomial time.
\end{lemma}

The $2$-approximation algorithm can easily be obtained from Lemma~\ref{lem:up_rounding}
as follows. For every $\ell =uv \in L\setminus L^{up}$ with $\ell \in \supp{x}$
set $w = \nca{uv}$ and adapt the solution $x$ to obtain a solution $x'$
as follows. Set $x'_{uw} = x_{uw} + x_\ell, \,\, x'_{vw} = x_{vw} + x_\ell,\,\, x'_\ell = 0$ 
and $x'_{\ell'} = x_{\ell'}$ for all other links $\ell'$. This eliminates $\ell$ from the support 
of $x$ by incurring an additional cost of $c_\ell x_\ell$, without breaking the feasibility of the
fractional solution. By repeating this for all links in the support that are not up-links we obtain 
a feasible solution with cost at most $2c^\intercal x$, as desired (see Figure~\ref{fig:uprounding} 
for an illustration).
This transformation and Lemma~\ref{lem:up_rounding} imply the following proposition.

\begin{figure}[h]
\begin{center}
\begin{tikzpicture}[scale=0.4]

\begin{scope}

% \node (a1) at (2,1) [circle,draw=black!100,fill=black!10,thick,inner sep=1pt,minimum size=3mm] {};
% \node (a2) at (4,1) [circle,draw=black!100,fill=black!10,thick,inner sep=1pt,minimum size=3mm] {};
\node (a3) at (3,4) [circle,draw=black!100,fill=black!10,thick,inner sep=1pt,minimum size=3mm] {};
\node (a4) at (1,4) [circle,draw=black!100,fill=black!10,thick,inner sep=1pt,minimum size=3mm] {};
\node (a5) at (2,7) [circle,draw=black!100,fill=black!10,thick,inner sep=1pt,minimum size=3mm] {};
\node (a6) at (4,10) [circle,draw=black!100,fill=black!10,thick,inner sep=1pt,minimum size=3mm] {};
\node (a7) at (5,4) [circle,draw=black!100,fill=black!10,thick,inner sep=1pt,minimum size=3mm] {};
\node (a8) at (5,7) [circle,draw=black!100,fill=black!10,thick,inner sep=1pt,minimum size=3mm] {};
% \node (a9) at (8,1) [circle,draw=black!100,fill=black!10,thick,inner sep=1pt,minimum size=3mm] {};
% \node (a10) at (10,1) [circle,draw=black!100,fill=black!10,thick,inner sep=1pt,minimum size=3mm] {};
\node (a11) at (7,4) [circle,draw=black!100,fill=black!10,thick,inner sep=1pt,minimum size=3mm] {};
\node (a12) at (9,4) [circle,draw=black!100,fill=black!10,thick,inner sep=1pt,minimum size=3mm] {};
\node (a13) at (8,7) [circle,draw=black!100,fill=black!10,thick,inner sep=1pt,minimum size=3mm] {};
\node (a14) at (8,10) [circle,draw=black!100,fill=black!10,thick,inner sep=1pt,minimum size=3mm] {};
\node (a15) at (8,13) [circle,draw=black!100,fill=black!10,thick,inner sep=1pt,minimum size=3mm] {};
% \node (a16) at (11,4) [circle,draw=black!100,fill=black!10,thick,inner sep=1pt,minimum size=3mm] {};
% \node (a17) at (13,1) [circle,draw=black!100,fill=black!10,thick,inner sep=1pt,minimum size=3mm] {};
% \node (a18) at (13,4) [circle,draw=black!100,fill=black!10,thick,inner sep=1pt,minimum size=3mm] {};
% \node (a19) at (11,7) [circle,draw=black!100,fill=black!10,thick,inner sep=1pt,minimum size=3mm] {};
\node (a20) at (14,7) [circle,draw=black!100,fill=black!10,thick,inner sep=1pt,minimum size=3mm] {};
\node (a21) at (12,10) [circle,draw=black!100,fill=black!10,thick,inner sep=1pt,minimum size=3mm] {};
% \node (a22) at (15,1) [circle,draw=black!100,fill=black!10,thick,inner sep=1pt,minimum size=3mm] {};
% \node (a23) at (17,1) [circle,draw=black!100,fill=black!10,thick,inner sep=1pt,minimum size=3mm] {};
% \node (a24) at (16,4) [circle,draw=black!100,fill=black!10,thick,inner sep=1pt,minimum size=3mm] {};

% \draw [-,black,very thick] (a1) to node [black] {} (a3);
% \draw [-,black,very thick] (a2) to node [black] {} (a3);
\draw [-,black,very thick] (a3) to node [black] {} (a5);
\draw [-,black,very thick] (a4) to node [black] {} (a5);
\draw [-,black,very thick] (a5) to node [black] {} (a6);
\draw [-,black,very thick] (a7) to node [black] {} (a8);
% \draw [-,black,very thick] (a9) to node [black] {} (a12);
% \draw [-,black,very thick] (a10) to node [black] {} (a12);
\draw [-,black,very thick] (a8) to node [black] {} (a6);
\draw [-,black,very thick] (a11) to node [black] {} (a13);
\draw [-,black,very thick] (a12) to node [black] {} (a13);
\draw [-,black,very thick] (a13) to node [black] {} (a14);
\draw [-,black,very thick] (a6) to node [black] {} (a15);
\draw [-,black,very thick] (a14) to node [black] {} (a15);
% \draw [-,black,very thick] (a16) to node [black] {} (a19);
% \draw [-,black,very thick] (a17) to node [black] {} (a18);
% \draw [-,black,very thick] (a18) to node [black] {} (a20);
% \draw [-,black,very thick] (a19) to node [black] {} (a21);
\draw [-,black,very thick] (a21) to node [black] {} (a15);
% \draw [-,black,very thick] (a22) to node [black] {} (a24);
% \draw [-,black,very thick] (a23) to node [black] {} (a24);
% \draw [-,black,very thick] (a24) to node [black] {} (a20);
\draw [-,black,very thick] (a20) to node [black] {} (a21);

%links
\draw [-,blue,very thick,dotted,out=0,in=-120] (a12) to node [black] {} (a20);
\draw [-,blue,very thick,dotted,out=-45,in=-135] (a4) to node [black] {} (a7);

\draw [-,blue,very thick,dashed,out=75,in=-75] (a12) to node [black] {} (a15);
\draw [-,blue,very thick,dashed,out=150,in=-60] (a20) to node [black] {} (a15);
\draw [-,blue,very thick,dashed,out=90,in=-150] (a4) to node [black] {} (a6);
\draw [-,blue,very thick,dashed,out=120,in=-90] (a7) to node [black] {} (a6);
\end{scope}

\end{tikzpicture}
\end{center}
\caption{An illustration of up-links, and the simple $2$-approximation
algorithm. Two regular links (dotted lines) and their corresponding pairs of up-links
(dashed lines).}
\label{fig:uprounding}
\end{figure}
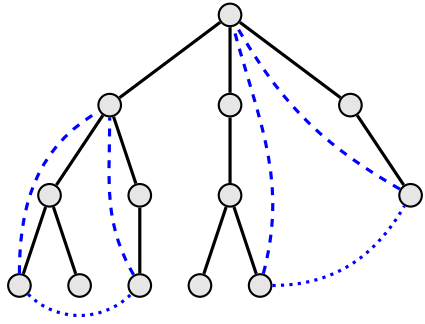

\begin{proposition}\label{prop:simple_rounding}
 There is a polynomial algorithm that given an instance $(G,L,c)$ of WTAP 
and a fractional solution $x\in \mathbb{R}^L_{\geq 0}$ to the natural LP 
returns a solution $S\subseteq L$ for the WTAP instance with cost
$c(S) \leq 2 c^\intercal x$.
\end{proposition}

\paragraph{Star-shaped WTAP and edge-covers.}

% An important special case of WTAP is the
% \emph{edge cover problem}. 
Given a graph $G$, an \emph{edge cover} is 
a subset $S\subseteq E[G]$ of the edges of the graph that is incident to every 
node of $G$, i.e.\ such that for every $u\in V[G]$, there exists an 
edge $e\in S$ such that $u\in e$.  A \emph{fractional edge cover}
is a vector $x\in \mathbb{R}_{\geq 0}^{E[G]}$ with the property that
$x(\{e\in E[G] \,\mid\, u\in e\}) \geq 1$ for every $u\in V[G]$. 
% In other words, a fractional
% edge cover is an assignment of positive values to edges, so that the 
% total value assigned to the edges incident to any node is at least $1$.

The minimum-cost edge cover problem is the problem of finding 
minimum-cost set of edges that is an edge cover of $G$. It is well-know 
that the problem is solvable in polynomial time. More importantly for our
purposes, the problem can be used to model a special case of WTAP, which 
we define next. 

\begin{definition}\label{def:star-shaped}
An instance $(G,L,c)$ of WTAP is called \emph{star-shaped} if there exists
a node $r\in V[G]$ such that $r$ is incident to $P_\ell$ for every link $\ell\in L$.
A node satisfying the latter condition is called a \emph{hub}.
\end{definition}

The following property of star-shaped instances is important to establish
the connection to edge covers. We call an edge in a tree a \emph{leaf edge}
if it is incident to a leaf in the tree. 

\begin{lemma}\label{lem:star-shaped}
 A solution $S\subseteq L$ for a star-shaped instance $(G,L,c)$ of WTAP
is feasible if and only if $S$ covers all leaf edges.
\end{lemma}

% \begin{proof}
%  The ``only if'' direction is trivial. To prove the ``if'' direction
% consider a solution $S\subseteq L$ that covers all leaf edges. We prove
% that it is feasible for the WTAP instance. Consider any edge $e\in E[G]$. We
% show that it is covered by $S$. Fix any hub $r\in V[G]$.
% There exists some leaf $u\in V[G]$ of the tree such that $e$ lies on 
% the $u$-$r$ path in $G$. Consider the link $\ell\in L$ that covers the 
% leaf edge of $u$. Since the instance is star-shaped, the path $P_\ell$
% is incident to $r$, and hence $e\in P_\ell$, as $P_\ell$ contains all
% edges on the $u$-$r$ path in $G$. It follows that $e$ is covered.
% \end{proof}

Lemma~\ref{lem:star-shaped} directly implies that in star-shaped instances, 
all links that are not incident to any leaf can be removed from $L$, as
they are redundant in any feasible solution. We hence assume that every
link in $L$ touches at least one leaf.

\begin{figure}[h]
\begin{center}
\begin{tikzpicture}[scale=0.4]

\begin{scope}

\node (a1) at (3,0) [circle,draw=black!100,fill=black!10,thick,inner sep=1pt,minimum size=3mm] {\scriptsize $1$};
\node (a2) at (5,3) [circle,draw=black!100,fill=black!10,thick,inner sep=1pt,minimum size=3mm] {};
\node (a3) at (2,6) [circle,draw=black!100,fill=black!10,thick,inner sep=1pt,minimum size=3mm] {\scriptsize $7$};
\node (a4) at (3,9) [circle,draw=black!100,fill=black!10,thick,inner sep=1pt,minimum size=3mm] {};
\node (a5) at (0,11) [circle,draw=black!100,fill=black!10,thick,inner sep=1pt,minimum size=3mm] {\scriptsize $6$};
\node (a6) at (7,0) [circle,draw=black!100,fill=black!10,thick,inner sep=1pt,minimum size=3mm] {\scriptsize $2$};
\node (a7) at (9,3) [circle,draw=black!100,fill=black!10,thick,inner sep=1pt,minimum size=3mm] {};
\node (a8) at (7,6) [circle,draw=black!100,fill=black!100,thick,inner sep=1pt,minimum size=4mm] {};
\node (a9) at (6,9) [circle,draw=black!100,fill=black!100,thick,inner sep=1pt,minimum size=4mm] {};
\node (a10) at (9,11) [circle,draw=black!100,fill=black!10,thick,inner sep=1pt,minimum size=3mm] {};
\node (a11) at (10,0) [circle,draw=black!100,fill=black!10,thick,inner sep=1pt,minimum size=3mm] {\scriptsize $3$};
\node (a12) at (12,3) [circle,draw=black!100,fill=black!10,thick,inner sep=1pt,minimum size=3mm] {};
\node (a14) at (12,6) [circle,draw=black!100,fill=black!10,thick,inner sep=1pt,minimum size=3mm] {\scriptsize $4$};
\node (a15) at (12,9) [circle,draw=black!100,fill=black!10,thick,inner sep=1pt,minimum size=3mm] {\scriptsize $5$};

\draw [-,black,very thick] (a1) to node [black] {} (a2);
\draw [-,black,very thick] (a2) to node [black] {} (a6);
\draw [-,black,very thick] (a2) to node [black] {} (a8);
\draw [-,black,very thick] (a3) to node [black] {} (a4);
\draw [-,black,very thick] (a4) to node [black] {} (a5);
\draw [-,black,very thick] (a4) to node [black] {} (a9);
\draw [-,black,very thick] (a7) to node [black] {} (a8);
\draw [-,black,very thick] (a8) to node [black] {} (a9);
\draw [-,black,very thick] (a9) to node [black] {} (a10);
\draw [-,black,very thick] (a7) to node [black] {} (a11);
\draw [-,black,very thick] (a7) to node [black] {} (a12);
\draw [-,black,very thick] (a12) to node [black] {} (a14);
\draw [-,black,very thick] (a10) to node [black] {} (a15);

%links
\draw [-,blue,very thick,dashed] (a14) to node [black] {} (a15);
\draw [-,blue,very thick,dashed,out=-70,in=80] (a10) to node [black] {} (a11);
\draw [-,blue,very thick,dashed] (a2) to node [black] {} (a3);
\draw [-,blue,very thick,dashed] (a5) to node [black] {} (a11);
\draw [-,blue,very thick,dashed,out=90,in=220] (a6) to node [black] {} (a15);
\draw [-,blue,very thick,dashed,out=135,in=-90] (a1) to node [black] {} (a5);
\draw [-,blue,very thick,dashed,out=150,in=-80] (a6) to node [black] {} (a3);
\end{scope}

\begin{scope}[xshift=20cm,scale=0.8]
\node (c1) at (2,2) [circle,draw=black!100,fill=black!10,thick,inner sep=1pt,minimum size=3mm] {\scriptsize $1$};
\node (c2) at (6,0) [circle,draw=black!100,fill=black!10,thick,inner sep=1pt,minimum size=3mm] {\scriptsize $2$};
\node (c3) at (10,2) [circle,draw=black!100,fill=black!10,thick,inner sep=1pt,minimum size=3mm] {\scriptsize $3$};
\node (c4) at (12,5) [circle,draw=black!100,fill=black!10,thick,inner sep=1pt,minimum size=3mm] {\scriptsize $4$};
\node (c5) at (12,9) [circle,draw=black!100,fill=black!10,thick,inner sep=1pt,minimum size=3mm] {\scriptsize $5$};
\node (c6) at (0,9) [circle,draw=black!100,fill=black!10,thick,inner sep=1pt,minimum size=3mm] {\scriptsize $6$};
\node (c7) at (0,5) [circle,draw=black!100,fill=black!10,thick,inner sep=1pt,minimum size=3mm] {\scriptsize $7$};

\node (d) at (6,12) [circle,draw=black!100,fill=black!100,thick,inner sep=1pt,minimum size=4mm] {};

\draw [-,blue,very thick] (c1) to node [black] {} (c6); 
\draw [-,blue,very thick] (c2) to node [black] {} (c7); 
\draw [-,blue,very thick] (c7) to node [black] {} (d); 
\draw [-,blue,very thick] (c3) to node [black] {} (c6); 
\draw [-,blue,very thick] (c3) to node [black] {} (d); 
\draw [-,blue,very thick] (c4) to node [black] {} (c5); 
\draw [-,blue,very thick] (c5) to node [black] {} (c2); 

\draw [-,blue,very thick] (d) .. controls +(1,1)  and +(1,0) .. +(0,2);
\draw [-,blue,very thick] (d) .. controls +(-1,1)  and +(-1,0) .. +(0,2);
%  .. controls +(-0.5,-1.5) .. (d);
\end{scope}

\end{tikzpicture}
\end{center}
\caption{An illustration of the transformation into an edge cover problem.
Left: A star shaped instance with hubs indicated as full nodes. Right: The 
corresponding edge cover instance.}\label{fig:starshaped}
\end{figure}
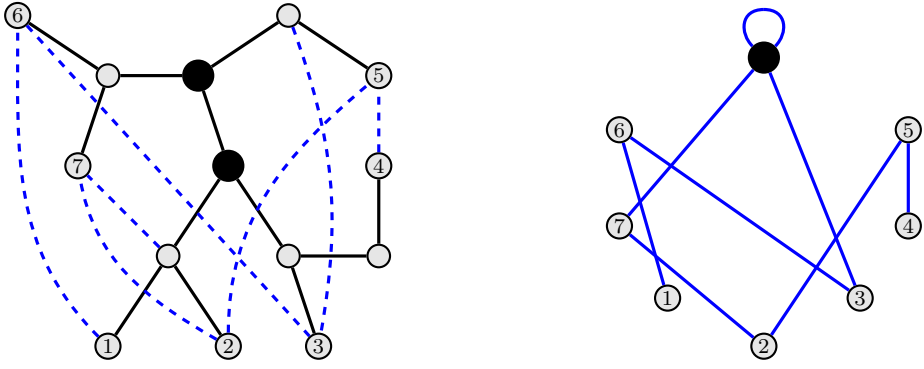

Lemma~\ref{lem:star-shaped} also implies that WTAP with 
star-shaped instances $(G,L,c)$ is equivalent to the edge cover problem: 
Construct a graph $H$ whose node set is $V[H] = U \cup \{r\}$, where $U$ is the set 
of leave in $G$, and $r$ is any hub. The set of edges $E[H]$
contains one edge $e^\ell$ for every link $\ell\in L$ with cost $c_\ell$. 
If $\ell$ connects two leaves $u,v\in U$, set $e^\ell = \{u,v\}$. 
Otherwise, if $\ell$ connects a leaf $u\in U$ to an internal node,
set $e^\ell = \{u,r\}$. 
Finally, we add a single self-loop $e_0$ connected to $r$ with cost zero. This
self-loop removes the need to incur cost for covering the node $r$.

It is immediate that edge covers in $H$ and feasible solutions to WTAP
instance $(G,L,c)$ are in one-to-one correspondence: $C\subseteq E[H]$
is an edge cover if and only if $\{\ell\in L\,\mid\, e^\ell\in C\}$ is
a feasible WTAP solution. Figure~\ref{fig:starshaped} illustrates the latter
transformation.

The most important implication from the latter correspondence 
is that one can round a fractional solution
to the natural LP relaxation of a star-shaped WTAP instance with a loss
of factor $\frac{4}{3}$. The following classical result on the integrality 
gap of the fractional edge cover polytope and the latter transformation 
imply this result, which we summarize in Proposition~\ref{prop:star-shaped}.

\begin{lemma}\label{lem:edge-cover-IG}
There is a polynomial algorithm that given a graph $G$, 
a cost vector $c\in \mathbb{R}^{E[G]}_{\geq 0}$ and
a fractional edge cover $x\in \mathbb{R}^{E[G]}_{\geq 0}$ computes
an edge cover $C\subseteq E[G]$ with cost $c(C) \leq \frac{4}{3}c^\intercal x$.
\end{lemma}

\begin{proposition}\label{prop:star-shaped}
 There is a polynomial algorithm that given a star-shaped WTAP
instance $(G,L,c)$ and a fractional solution $x\in \mathbb{R}^L_{\geq 0}$ of
the natural LP for this instance computes a feasible solution $S\subseteq L$
for the WTAP instance with cost $c(S) \leq \frac{4}{3} c^\intercal x$.
\end{proposition}

\section{A Rounding Algorithm for WTAP (Proof of Theorem~\ref{thm:wtap})}\label{sec:wtap}

We present next an algorithm that rounds an optimal solution
$x\in \mathbb{R}_{\geq 0}^L$ to $\mathrm{LP}_{\gamma}$, for a 
constant $\gamma$ that depends on the accuracy $\epsilon$ and the 
maximum weight $M$. We determine $\gamma$ explicitly later. To this 
end we outline the general strategy.

We round the solution in two phases. In the first phase 
we break the tree $G$ into a union of simpler subtrees, and equip
each subtree with a feasible fractional WTAP solutions, 
derived from $x$ by using a simple splitting operation. 
Each fractional solution will only use links connecting nodes 
in its corresponding subtree, allowing us to treat each subtree
and its corresponding solution separately. To bound the cost
of the decomposition it is first necessary to guarantee 
that no edge in the tree is over-covered, i.e.\ covered in $x$
by a fraction larger than some constant $a = a(\epsilon)$. This is achieved
by rounding an appropriately scaled version of the fractional solution
$x$, using Proposition~\ref{prop:simple_rounding}, and contracting 
the obtained set of links.

In the second phase the fractional solutions in each subtree are rounded
to integral solutions using two different procedures. Here, an important
structural property of each subtree-solution pair in the decomposition is used, 
which is called \emph{simplicity}, and is parametrized by an integer $\beta$.
The links of each $\beta$-simple pair are partitioned into two types.
Then, depending on which type of links dominates the cost in the 
fractional solution, one of the two rounding procedures is used
to obtain an integral solution, as each rounding procedure achieves 
a good approximation with respect to one type of links. The bundle
constraints are exploited in one of the rounding procedures, while
the other rounding procedure uses the fact that instances corresponding to
$\beta$-simple pairs are close to being star-shaped. Finally, the union of all solutions from 
all subtrees in the decomposition is returned as the solution.

\subsection{Phase One: Decomposition}

The decomposition of $G$ and $x$ is obtained in two steps, which we describe
hereafter.

\subsubsection{Contraction of Heavily Covered Edges}

In the first step we select a low-cost set of links to 
cover all edges of $G$ that are covered by a total weight of
at least some constant $a = a(\epsilon)$. Concretely, define
the set of edges that are \emph{heavily covered} as
$$
E^{h} = \left\{ e\in E \,\mid\, x(\cov{e}) \geq \frac{2}{\epsilon} \right\},
$$
namely $E^h$ are the edges that are covered by $x$ with links with
a mass of at least $\frac{2}{\epsilon}$. To obtain the desired set of links 
we contract the edges $E\setminus E^{h}$ to obtain the subtree $G^{h}$ 
of $G$ whose edge set is $E^{h}$. Now, the solution $x$ covers \emph{every} edge
in $G^h$ by a fraction of at least $\frac{2}{\epsilon}$. 
It follows that $y = \frac{\epsilon}{2} \cdot x$ is a feasible solution to 
the natural LP relaxation of the WTAP instance on $G^{h}$, and hence
Proposition~\ref{prop:simple_rounding} can be applied. The result 
is a set of links $L_0 \subseteq L$ that cover all edges of $G^h$ and
has cost
$
c(L_0) \leq 2 c^\intercal y = \epsilon c^\intercal x,
$
i.e.\ its cost is only an $\epsilon$-fraction of the cost of $x$. The
links $L_0$ are included in the solution, so it henceforth 
remains to cover all edges in $E[G] \setminus E^h$. 

Let $\bar G$ be the tree obtained by contracting the edges in $E^h$. Note that $E[\bar G] = E[G] \setminus E^h$, so the edges of $\bar G$ are exactly the ones that we still need to cover.
The key property of the solution $x$, interpreted in the new tree
$\bar G$, is the following \emph{thin coverage property}, which 
states that for every $e\in E[\bar G]$ it holds that
$$
x(\cov{e}) \leq \frac{2}{\epsilon} = O\left(\frac{1}{\epsilon}\right),
$$
a property that we crucially exploit to bound the cost of the 
decomposition step. 

We note that $x$ might not be a solution to $\mathrm{LP}_{\gamma}$ 
on $\bar G$, for the same $\gamma$ that we used to 
obtain $x$ as a solution for $G$. This is due to the fact
that $\gamma$-bundles in $\bar G$ might contain compound nodes,
obtained by contracting some links in $L_0$, and hence they might
not represent $\gamma$-bundles of $G$. However, as we will later show, 
$x$ maintains enough of the structure given by the 
$\gamma$-bundle constraints to bound the integrality gap. To make 
these arguments precise later on we keep track
of which nodes of the new tree $\bar G$ are compound nodes, namely nodes
that represent more than one original node of the tree $G$, and were 
created by contracting some links in $L_0$. These nodes are denoted
by $V^{cp}$ and called \emph{early compound nodes} 
to distinguish them from compound nodes obtained due to later contractions
that we perform in the algorithm. Furthermore, for $u\in V^{cp}$, we denote by 
$s_u \in \mathbb{Z}_{\geq 0}$ the total cost of links of $L_0$ that were
contracted to obtain the early compound node $u$. For non-compound nodes
$u\in V[\bar G]\setminus V^{cp}$ we set $s_u = 0$. Since links
cover paths in $G$, one link of $L_0$ cannot contribute to the formation
of more than one early compound node, so we have 
$$
\sum_{u\in V[\bar G]} s_u = \sum_{u\in V^{cp}} s_u = c(L_0).
$$

%%%%%

\subsubsection{Decomposition}

In the next step the algorithm decomposes
the instance into simpler instances, by breaking the tree at 
certain edges using a simple splitting operation. Each 
part of the obtained decomposition is a pair $(T,z)$, where 
$T$ is a subtree of $\bar G$, and $z$ is a fractional 
solution for the WTAP instance restricted to $T$.
The decomposition is obtained by an iterative greedy procedure
that employs the following operation, that we call \emph{splitting}. 

\begin{definition}
Let $(G,L,c)$ be a WTAP instance and
let $z\in \mathbb{R}^L_{\geq 0}$. Let $e = \{u,v\} \in E$ be any edge. 
Let $G^u$ and $G^v$ be the trees obtained by removing $e$ from $G$, where
$G^u$ is the tree that contains $u$. The \emph{splitting of $z$
at $e$} produces two vectors $z^u\in \mathbb{R}^L_{\geq 0}$ and $z^v\in \mathbb{R}^L_{\geq 0}$ defined
as follows. We define $z^u$; $z^v$ is defined symmetrically. For
$\ell= pq \in L$ set
$$
z^u_\ell = \begin{cases} 
				     z_\ell & \mbox{if } \,\,\,  p,q\in V[G^u]\setminus\{u\} \\ 
				     0 & \mbox{if } \,\,\, \{p,q\}\cap  V[G^v] \neq  \emptyset \\
				     z_\ell + \sum_{\ell' \in \cov{e},\,\, q\in \ell'} z_{\ell'} & \mbox{if } \,\,\, p=u, \, q\in V[G^u].
				\end{cases}
$$
Note that $\supp{z^u} \subseteq V[G^u] \times V[G^u]$ 
and $\supp{z^v} \subseteq V[G^v] \times V[G^v]$.
\end{definition}

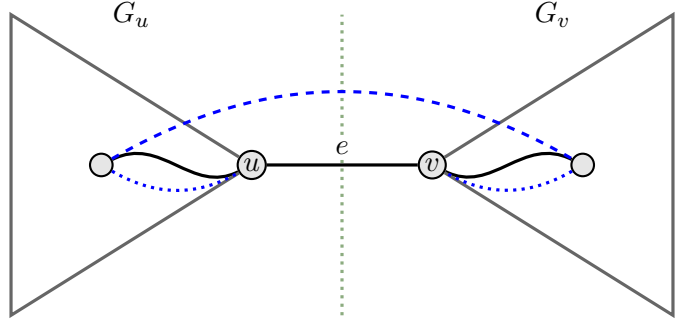
\begin{figure}[h]
\begin{center}
\begin{tikzpicture}[scale=0.4]

\begin{scope}

\node (u) at (8,5) [circle,draw=black!100,fill=black!10,thick,inner sep=1pt,minimum size=3mm] {$u$};
\node (v) at (14,5) [circle,draw=black!100,fill=black!10,thick,inner sep=1pt,minimum size=3mm] {$v$};

\node (a) at (3,5) [circle,draw=black!100,fill=black!10,thick,inner sep=1pt,minimum size=3mm] {};
\node (b) at (19,5) [circle,draw=black!100,fill=black!10,thick,inner sep=1pt,minimum size=3mm] {};

\draw [-,black,very thick] (a) .. controls (5,6)  and (6,4) .. (u);
\draw [-,black,very thick] (b) .. controls (17,6)  and (16,4) .. (v);

\node (gu) at (4,10) [circle,draw=black!0,fill=black!0,thick,inner sep=1pt,minimum size=3mm] {$G_u$};
\node (gv) at (18,10) [circle,draw=black!0,fill=black!0,thick,inner sep=1pt,minimum size=3mm] {$G_v$};

\draw [-,dotted,draw=darkgreen!50,very thick] (11,0) -- (11,5.3);
\draw [-,dotted,draw=darkgreen!50,very thick] (11,6) -- (11,10);

\draw [-,draw=black!60,very thick] (u) -- (0,10) -- (0,0) -- (u);
\draw [-,draw=black!60,very thick] (v) -- (22,10) -- (22,0) -- (v);

\draw [-,blue,very thick,dashed,out=30,in=150] (a) to node [black] {} (b);
\draw [-,blue,very thick,dotted,out=-30,in=-150] (a) to node [black] {} (u);
\draw [-,blue,very thick,dotted,out=-150,in=-30] (b) to node [black] {} (v);

\draw [-,black,very thick] (u) to node [black,above] {$e$} (v);
%links
% \draw [-,blue,very thick,dotted,out=0,in=-120] (a12) to node [black] {} (a20);
\end{scope}

\end{tikzpicture}
\end{center}
\caption{The splitting of $z$ at $e$. The fractional value of the link crossing
the cut (dashed) is added to $z^u$-value and $z^v$-value of the left and 
right shadows (dotted) of the link, respectively. 
}
\label{fig:splitting}
\end{figure}

Figure~\ref{fig:splitting} illustrates the splitting operation.
It is easy to see that if $y$ is a feasible fractional solution to the
natural LP for the WTAP instance $(G, L, c)$ and $e = \{u,v\} \in E[G]$,
then the splitting of $y$ at $e$ produces two feasible fractional 
solutions for the natural LP: $y^u$ is a feasible solution for $G^u$ and $y^v$ is 
a feasible solution for $G^v$. Furthermore, the total weight of 
$y^u$ and $y^v$ is easy to express in terms of the weight of $y$ and
total weight in $y$ of links covering $e$:
$$
c^\intercal y^u + c^\intercal y^v = c^\intercal y + \sum_{\ell \in \cov{e}} c_\ell y_\ell.
$$
The latter equality follows from shadow completeness and the fact that for 
each link $\ell \in L$, $y_\ell$ is either counted in $y^u$, or $y^v$, and it is counted
in both if and only if $\ell \in \cov{e}$. Assuming that $y$ also satisfies the
thin coverage property (recall that the optimal fractional solution $x$ satisfies this property), 
then the additive term in the latter expression can also be easily bounded in terms of 
the maximum cost of any link and the parameter of the thin coverage property:
$$
\sum_{\ell \in \cov{e}} c_\ell y_\ell \leq M \cdot y(\cov{e}) \leq \frac{2M}{\epsilon}.
$$
Furthermore, clearly $y^u$ and $y^v$ also satisfy the thin coverage property
in the corresponding trees $G^u$ and $G^v$, so the splitting operation does not violate
this property.

Next, we use the splitting operation to decompose the tree $\bar G$ into simpler
trees. We employ a greedy procedure that maintains 
a set of pairs $\mathcal{T}$, each comprising a subtree of $\bar G$ and a 
fractional solution of the natural LP for this subtree. We initialize by 
setting $\mathcal{T} = \{(\bar G, x)\}$. At each iteration, the algorithm 
chooses an arbitrary pair $(T,z)\in \mathcal{T}$ and checks if it contains
a \emph{thin edge}, a notion that we define next.

\begin{definition}\label{def:thin-edge}
Let $T$ be a subtree of $\bar G$, let $z\in \mathbb{R}_{\geq 0}^L$ and $\alpha \in \mathbb{R}_{\geq 0}$.
An edge $e = \{u,v\}\in E[T]$ is called \emph{$\alpha$-thin with respect to $z$}
if the total cost of links connecting nodes in $V[G^u]$, and the total cost
of links that connect two nodes in $V[G^v]$ is at least $\alpha$, namely if
$$
\sum_{\ell\in L, \ell \in V[G^q]\times V[G^q]} c_\ell z_\ell \geq \alpha \quad\quad \text{for} \quad q=u,v.
$$
\end{definition}

Formally, the algorithm selects any $\alpha(M,\epsilon)$-thin edge with respect to $z$ for
$$
\alpha(M,\epsilon) = \frac{4M}{\epsilon^2},
$$
removes $(T,z)$ from $\mathcal{T}$, adds to $\mathcal{T}$ the pairs
$(T^u, z^u)$ and $(T^v, z^v)$, obtained from $(T,z)$ by splitting of $z$ at $e$
and proceeds to the next iteration. 
If no $\alpha(M,\epsilon)$-thin edge is found, the algorithm
reports $(T,z)$ as part of the final decomposition of $(\bar G, x)$, removes 
it from $\mathcal{T}$ and proceeds to the next iteration. 
After at most $|V[\bar G]| - 1$ iterations  $\mathcal{T}$ is empty, 
at which stage the algorithm terminates and returns the full decomposition 
$(T^1, z^1), \cdots, (T^k, z^k)$ of $(\bar G, x)$. 

The decomposition produced by the latter algorithm has several useful 
properties, which we state and prove next. First, each pair 
$(T^j, z^j)$ in the decomposition can be seen as a WTAP instance 
$(T^j, L^j, c)$, where $L^j = L \cap \left(V[T^j]\times V[T^j]\right)$,
and for which $z^j$ is a feasible fractional solution of the 
natural LP (without the bundle constraints).
Formally, $z^j\in \mathbb{R}_{\geq 0}^L$, but from the way that the splitting
operation works, we have $\supp{z^j} \subseteq L^j$, so we can indeed interpret
is as a solution of the instance $(T^j, L^j, c)$.

Furthermore, the decomposition satisfies the following property, which will be important 
later, when we exploit the bundle constraints in the bundle LP. Informally
speaking, the way weight is shifted around does not 
decrease the total fractional cost used to cover any subset of edges, provided
that this subset is contained in some part of the decomposition.

\begin{lemma}\label{lem:weight-presenving}
 Let $j\in [k]$ and let $F\subseteq E[T^j]$ be any set of edges. Then 
$$
\sum_{\ell \in \cov{F}} c_\ell z^j_\ell \geq \sum_{\ell \in \cov{F}} c_\ell x_\ell.
$$
\end{lemma}

% \begin{proof}
%  The proof follows easily from the fact that in any splitting operation performed
% in the algorithm, whenever the fractional assignment of some link $\ell$ covering 
% some edges in $F$ is decreased to zero, the fractional assignment of
% a shadow $\ell'$ of $\ell$ with the same cost is increases by the same fraction.
% \end{proof}

Next, we show that the total increase in cost incurred by the decomposition is very small.

\begin{lemma}\label{lem:weight-increase-decomp}
 Let $(T^1, z^1), \cdots, (T^k, z^k)$ be a decomposition of $(\bar G, x)$ produced
by the greedy procedure. Then
$$
\sum_{i\in [k]} c^\intercal z^i \leq (1 + \epsilon) c^\intercal x.
$$
\end{lemma}

% \begin{proof}
% Observe that the total number of splittings in the greedy procedure is 
% equal to the number of parts in the final decomposition minus one, namely
% $k-1$. Due to the thin coverage property, we know that in each splitting the
% cost is increased by a total of at most $\frac{2M}{\epsilon}$, so the
% total increase satisfies 
% $$
% \sum_{i\in [k]} c^\intercal z^i - c^\intercal x \leq \frac{2kM}{\epsilon}.
% $$ 
% Now, since we only perform splitting at $\alpha(M,\epsilon)$-thin edges, 
% each pair $(T^j,z^j)$ in the decomposition satisfies 
% $c^\intercal z^j \geq \frac{4M}{\epsilon^2}$, hence
% $$
% \sum_{i\in [k]} c^\intercal z^i \geq \frac{4kM}{\epsilon^2}, 
% $$
% which implies
% $\frac{\sum_{i\in [k]} c^\intercal z^i - c^\intercal x}{c^\intercal x} \leq \epsilon,$ as desired.
% % \frac{\frac{2kM}{\epsilon}}{\frac{4kM}{\epsilon^2} - \frac{2kM}{\epsilon}} \leq
% % ...
% % $$
% \end{proof}

In addition, the trees in the decomposition have a convenient structure that we define next.

\begin{definition}\label{def:simple-tree}
Let $\beta\in \mathbb{Z}_{\geq 1}$. Call a pair $(T,z)$ \emph{$\beta$-simple} if there
exists a node $u\in V[T]$ the removal of which results in a forest with trees
$K_1, \cdots, K_t$, such that for each $j\in [t]$
\begin{itemize}
 \item $\sum_{\ell \in L, \,\, \ell \in V[K_j]\times V[K_j]} c_\ell z_\ell \leq \beta$, and 
 \item $K_j$ has at most $\beta$ leaves.
\end{itemize}
A node $u$ that leads to such a partition is called a \emph{$\beta$-center} of $(T,z)$.
\end{definition}

\begin{figure}[h]
\begin{center}
\begin{tikzpicture}[scale=0.4]

\begin{scope}

\node (a1) at (2,6) [circle,draw=black!100,fill=black!50,thick,inner sep=1pt,minimum size=3mm] {};
\node (a2) at (2,10) [circle,draw=black!100,fill=black!10,thick,inner sep=1pt,minimum size=3mm] {};
\node (a3) at (5,6) [circle,draw=black!100,fill=black!50,thick,inner sep=1pt,minimum size=3mm] {};
\node (a4) at (5,10) [circle,draw=black!100,fill=black!10,thick,inner sep=1pt,minimum size=3mm] {};
\node (a5) at (5,14) [circle,draw=black!100,fill=black!10,thick,inner sep=1pt,minimum size=3mm] {};
\node (a6) at (7,2) [circle,draw=black!100,fill=black!50,thick,inner sep=1pt,minimum size=3mm] {};
\node (a7) at (7,6) [circle,draw=black!100,fill=black!10,thick,inner sep=1pt,minimum size=3mm] {};
\node (a8) at (8,10) [circle,draw=black!100,fill=black!10,thick,inner sep=1pt,minimum size=3mm] {};
\node (a9) at (9,6) [circle,draw=black!100,fill=black!50,thick,inner sep=1pt,minimum size=3mm] {};
\node (a10) at (11,6) [circle,draw=black!100,fill=black!50,thick,inner sep=1pt,minimum size=3mm] {};
\node (a11) at (12,10) [circle,draw=black!100,fill=black!10,thick,inner sep=1pt,minimum size=3mm] {};
\node (a12) at (13,6) [circle,draw=black!100,fill=black!50,thick,inner sep=1pt,minimum size=3mm] {};
\node (a13) at (16,2) [circle,draw=black!100,fill=black!50,thick,inner sep=1pt,minimum size=3mm] {};
\node (a14) at (16,6) [circle,draw=black!100,fill=black!10,thick,inner sep=1pt,minimum size=3mm] {};
\node (a15) at (16,10) [circle,draw=black!100,fill=black!10,thick,inner sep=1pt,minimum size=3mm] {};
\node (a16) at (14,14) [circle,draw=black!100,fill=black!10,thick,inner sep=1pt,minimum size=3mm] {};
\node (a17) at (15,17) [circle,draw=black!100,fill=black!100,thick,inner sep=1pt,minimum size=3mm] {};
\node (a18) at (22,2) [circle,draw=black!100,fill=black!50,thick,inner sep=1pt,minimum size=3mm] {};
\node (a19) at (23,6) [circle,draw=black!100,fill=black!10,thick,inner sep=1pt,minimum size=3mm] {};
\node (a20) at (24,2) [circle,draw=black!100,fill=black!50,thick,inner sep=1pt,minimum size=3mm] {};
\node (a21) at (23,10) [circle,draw=black!100,fill=black!10,thick,inner sep=1pt,minimum size=3mm] {};
\node (a22) at (26,6) [circle,draw=black!100,fill=black!50,thick,inner sep=1pt,minimum size=3mm] {};
\node (a23) at (27,10) [circle,draw=black!100,fill=black!10,thick,inner sep=1pt,minimum size=3mm] {};
\node (a24) at (28,6) [circle,draw=black!100,fill=black!50,thick,inner sep=1pt,minimum size=3mm] {};
\node (a25) at (25,14) [circle,draw=black!100,fill=black!10,thick,inner sep=1pt,minimum size=3mm] {};

\draw [-,black,very thick] (a1) to node [black,above] {} (a2);
\draw [-,black,very thick] (a3) to node [black,above] {} (a4);
\draw [-,black,very thick] (a2) to node [black,above] {} (a5);
\draw [-,black,very thick] (a4) to node [black,above] {} (a5);
\draw [-,black,very thick] (a6) to node [black,above] {} (a7);
\draw [-,black,very thick] (a7) to node [black,above] {} (a8);
\draw [-,black,very thick] (a9) to node [black,above] {} (a8);
\draw [-,black,very thick] (a8) to node [black,above] {} (a5);
\draw [-,black,very thick] (a5) to node [black,above] {} (a17);
\draw [-,black,very thick] (a10) to node [black,above] {} (a11);
\draw [-,black,very thick] (a12) to node [black,above] {} (a11);
\draw [-,black,very thick] (a13) to node [black,above] {} (a14);
\draw [-,black,very thick] (a14) to node [black,above] {} (a15);
\draw [-,black,very thick] (a15) to node [black,above] {} (a16);
\draw [-,black,very thick] (a11) to node [black,above] {} (a16);
\draw [-,black,very thick] (a16) to node [black,above] {} (a17);
\draw [-,black,very thick] (a18) to node [black,above] {} (a19);
\draw [-,black,very thick] (a20) to node [black,above] {} (a19);
\draw [-,black,very thick] (a19) to node [black,above] {} (a21);
\draw [-,black,very thick] (a22) to node [black,above] {} (a23);
\draw [-,black,very thick] (a24) to node [black,above] {} (a23);
\draw [-,black,very thick] (a21) to node [black,above] {} (a25);
\draw [-,black,very thick] (a23) to node [black,above] {} (a25);
\draw [-,black,very thick] (a25) to node [black,above] {} (a17);

\draw [-,blue,very thick,dotted,out=-90,in=-180] (a1) to node [black] {} (a6);
\draw [-,blue,very thick,dotted,out=70,in=-110] (a1) to node [black] {} (a5);
\draw [-,blue,very thick,dashed,out=110,in=-110] (a3) to node [black] {} (a4);
\draw [-,blue,very thick,dashed,out=-10,in=-170] (a4) to node [black] {} (a8);
\draw [-,blue,very thick,dotted] (a6) to node [black] {} (a9);
\draw [-,blue,very thick,dotted,out=-50,in=170] (a9) to node [black] {} (a13);
\draw [-,blue,very thick,dotted,out=-170,in=-10] (a13) to node [black] {} (a6);
\draw [-,blue,very thick,dashed,out=30,in=-140] (a10) to node [black] {} (a15);
\draw [-,blue,very thick,dashed,out=-10,in=-170] (a12) to node [black] {} (a14);
\draw [-,blue,very thick,dashed,out=-40,in=150] (a14) to node [black] {} (a18);
\draw [-,blue,very thick,dashed,out=50,in=-100] (a20) to node [black] {} (a22);
\draw [-,blue,very thick,dashed,out=80,in=-30] (a24) to node [black] {} (a25);

\node (k1) at (5,0) [circle,draw=black!0,fill=black!0,thick,inner sep=1pt,minimum size=3mm] {$K_1$};
\node (k2) at (13.5,0) [circle,draw=black!0,fill=black!0,thick,inner sep=1pt,minimum size=3mm] {$K_2$};
\node (kt) at (25,0) [circle,draw=black!0,fill=black!0,thick,inner sep=1pt,minimum size=3mm] {$K_t$};

\node (d1) at (19,0) [circle,draw=black!0,fill=black!0,thick,inner sep=1pt,minimum size=3mm] {$\boldsymbol{\cdots}$};
\node (d2) at (19,8) [circle,draw=black!0,fill=black!0,thick,inner sep=1pt,minimum size=3mm] {$\boldsymbol{\cdots}$};
\node (d3) at (19,14) [circle,draw=black!0,fill=black!0,thick,inner sep=1pt,minimum size=3mm] {$\boldsymbol{\cdots}$};

\draw [draw=black!40, very thick, dashed] (10,-1) -- (10,18);
\draw [draw=black!40, very thick, dashed] (17,-1) -- (17,18);
\draw [draw=black!40, very thick, dashed] (21,-1) -- (21,18);

\end{scope}

\end{tikzpicture}
\end{center}
\caption{A $4$-simple pair $(T,z)$ of an (unweighted) TAP instance. The full node
is a $4$-center. The shown links represent the support of $z$, with integral (dashed)
and half-integral (dotted) links. The number of leaves, as well as the total fractional
weight in each subtree is at most $4$.}
\label{fig:simplepair}
\end{figure}
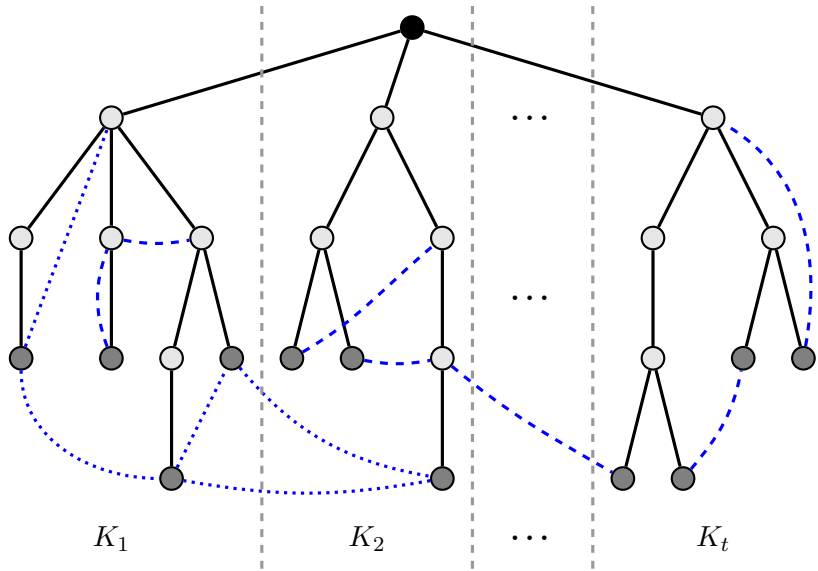

Figure~\ref{fig:simplepair} illustrates the definition of a $\beta$-simple pair.
Informally speaking, $(T,z)$ is $\beta$-simple (for small $\beta$) if by removing 
a singe node from $T$, one can break it into trees, where each tree has a few leaves, 
as well as low cost on links that are fully contained in the tree. We stress that
a tree of a $\beta$-simple pair is \emph{not} a tree that can be decomposed into constant size, 
or constant depth trees (as in~\cite{cohen2013}) by removing a single node. Instead, the 
number of leaves in each part is small, while the total number of nodes can be very large.

The following lemma states 
that every tree-solution pair in the obtained decomposition is $\beta$-simple for
$$
\beta(M, \epsilon) = \frac{48M}{\epsilon^2}.
$$

\begin{lemma}\label{lem:simple-trees}
 Let $(T^1, z^1), \cdots, (T^k, z^k)$ be a decomposition of $(\bar G, x)$ produced
by the greedy procedure. Then every pair in the decomposition is 
$\beta(\epsilon, M)$-simple for $\beta(\epsilon, M) = \frac{48M}{\epsilon^2}$.
\end{lemma}

This concludes the description of the decomposition and its properties.
However, before we proceed with the second phase of the algorithm, in 
which each part of the decomposition is rounded to an integral solution,
we need to take care of the following small technicality. 

While the union of the trees in the decomposition contains all nodes in $\bar G$,
it does not contain all edges in $\bar G$. More precisely, the $k-1$ edges 
used in the $k-1$ splitting operations that resulted in the final decomposition
are not part of any tree in the decomposition. Hence, we need to explain
how these edges are covered in the solution returned by the algorithm. 
Here, a trivial solution $L_1 \subseteq L$ containing an arbitrary covering link per edge
has a cost of at most 
%This can, however, be trivially done by taking any set of at most $k-1$ links 
%$L_1 \subseteq L$ with cost at most 
$(k-1)M$, which is at most an $\epsilon$-fraction of total fractional cost: 
Each one of the $k$ solutions $z^1, \cdots, z^k$ has, by the property of the decomposition,
a fractional cost of at least $\alpha(M,\epsilon) = \frac{4M}{\epsilon^2}$,
implying that the total fractional cost is at least $\frac{4Mk}{\epsilon^2}$, so
$c(L_1) \leq \epsilon \cdot \sum_{i\in [k]} c^\intercal z^i.$

\subsection{Phase 2: Rounding $\beta$-Simple Pairs}

The second phase of the rounding algorithm accepts a decomposition
$(T^1,z^1), \cdots, (T^k,z^k)$ of the pair $(\bar G, x)$ into 
$\beta(M,\epsilon)$-simple pairs, and outputs $k$ WTAP solutions
$S_1, \cdots, S_k$, one for each pair. The final output of the algorithm
is hence 
$
S^{\ALG} = L_0 \cup L_1 \cup S_1 \cup \dots \cup S_k,
$
which is feasible for the original WTAP instance. 
As we have shown, the contribution of $L_0 \cup L_1$
to the cost of the solution is $O(\epsilon) c^\intercal x$, so
its contribution can henceforth be neglected, as it entails an 
arbitrarily small loss in the approximation guarantee. Furthermore, 
Lemma~\ref{lem:weight-increase-decomp} guarantees that the total
cost of the solutions $z^1, \cdots, z^k$ is at most an 
$\epsilon$-fraction larger than that of $x$. Finally,
since $\supp{z^j} \subseteq V[T^j]\times V[T^j]$ holds for 
every $j\in [k]$, we can treat each pair $(T^j, z^j)$ as a separate
instance-solution pair, thus neglecting the dependencies between 
the pairs and presenting a rounding procedure for one such pair.

Let $(T,z)$ henceforth denote any pair in the decomposition.
We present two rounding procedures, each achieving a good 
approximation with respect to some part of the fractional solution $z$.
It is then easy to show that the approximaiton guarantee claimed
in Theorem~\ref{thm:wtap} is attained by reporting the solution
with the lower cost among the two solutions.

Interestingly, only one of the rounding procedures exploits the
bundle constraints in the LP. The other procedure only uses properties
of star-shaped solutions and Proposition~\ref{prop:star-shaped},
and provides a good approximation when the instance is \emph{close to being
star-shaped}. We present this procedure first. For what remains 
we fix any $\beta(M,\epsilon)$-center $r\in V[T]$, call it \emph{root},
and denote by $R^1, \cdots, R^m$ the set of trees that are obtained by 
removing $r$ from $T$.

\subsubsection{First Rounding Procedure: Nearly Star-Shaped Pairs}

We call a link $\ell \in V[T] \times V[T]$ a \emph{cross-link} if 
it connects two nodes in different trees among $R^1, \cdots, R^m$.
Observe that any cross-link $\ell$ has the root $r$ incident to
its path $P^T_\ell$. A link $\ell \in V[T] \times V[T]$ that is not a cross-link 
is called an \emph{in-link}. Note that all links $\ell$ with $r\in \ell$
are in-links.
% A link $\ell = uv$ is called 
% an \emph{up-link} if one node, say $u$ is on the path to root $r$
% of the other node, i.e.\ if $u$ is in the $v$-$r$ path in $T$.

Define $z^{cr}\in \mathbb{R}_{\geq 0}^L$ and $z^{in}\in \mathbb{R}_{\geq 0}^L$ 
to be the parts of $z$ that correspond to cross-links and in-links, 
respectively. Formally, $z^{cr}_\ell = z_\ell$ if $\ell$ is a cross-link 
and $z^{cr}_\ell = 0$, otherwise, and $z^{in} = z - z^{cr}$.
The following lemma proves the existence of a simple rounding algorithm
that for any $\lambda > 1$ produces a solution with cost at most $\frac{4\lambda}{3(\lambda - 1)}$
times the total cost of cross-links in $z$, albeit at a high cost in terms
of in-links.

\begin{lemma}\label{lem:cross-rounding}
 Let $\lambda > 1$ be any constant. Given $T$ and $z$, there is a 
polynomial time algorithm that produces a set of links $S\subseteq L$ 
that covers $E[T]$ with cost at most 
$$
c(S) \leq 2\lambda c^\intercal z^{in} + \frac{4\lambda}{3(\lambda - 1)} c^\intercal  z^{cr}.
$$
\end{lemma}

\begin{proof}
 The rounding algorithm works as follows. Denote by $E_\lambda$ the set of
edges in $E[T]$ that are covered by a fraction of at least $\frac{1}{\lambda}$
by in-links in $z$. In other words
$$
E_\lambda = \left\{ e\in E[T] \,\mid\, z^{in}(\cov{e}) \geq \frac{1}{\lambda} \right\}.
$$
Let $y = \lambda z^{in}$ be a fractional WTAP solution that covers every edge
in $E_\lambda$ completely (i.e.\ $y(\cov{e}) \geq 1$ for all $e\in E_\lambda$). 
Now, invoke Proposition~\ref{prop:simple_rounding} to produce
a solution $S^1 \subseteq L$ with cost
$$
c(S^1) \leq 2 c^\intercal y = 2\lambda c^\intercal z^{in}
$$ 
that covers all edges in $E_\lambda$. Contract all edges in $E_\lambda$ to obtain
a new tree $T'$ with edge set $E[T'] = E[T] \setminus E_\lambda$.

By definition of $E_\lambda$, and from feasibility of $z$, every edge in $E[T']$ is covered by at 
least a fraction $\frac{\lambda-1}{\lambda}$ with cross-links, namely
$$
z^{cr}(\cov{e}) \geq \frac{\lambda- 1}{\lambda}
$$
holds for all $e\in E[T']$, so $y' = \frac{\lambda}{\lambda - 1} z^{cr}$ is 
a fractional WTAP solution for the tree $T'$. 
Our aim is to cover all edges in $T'$
with cross-links contained in $L^{cr} = \supp{z^{cr}}$. Observe, that the WTAP
instance restricted to these links is feasible, by definition of $T'$, and 
furthermore, it is \emph{star-shaped}, as for any cross-link $\ell \in L^{cr}$,
the path $P^{T'}_\ell$ is incident to the root $r$. We can now
apply Proposition~\ref{prop:star-shaped}
to construct a solution $S^2 \subseteq L$ covering all edges of $T'$ with cost
$$
c(S^2) \leq \frac{4}{3} c^\intercal y' = \frac{4\lambda}{3(\lambda - 1)} c^\intercal z^{cr}.
$$
Since $S = S^1\cup S^2$ comprises a feasible solution to the WTAP instance on $T$, the lemma
is proved. Figure~\ref{fig:star-shaped-rounding} illustrates the proof.
\end{proof}

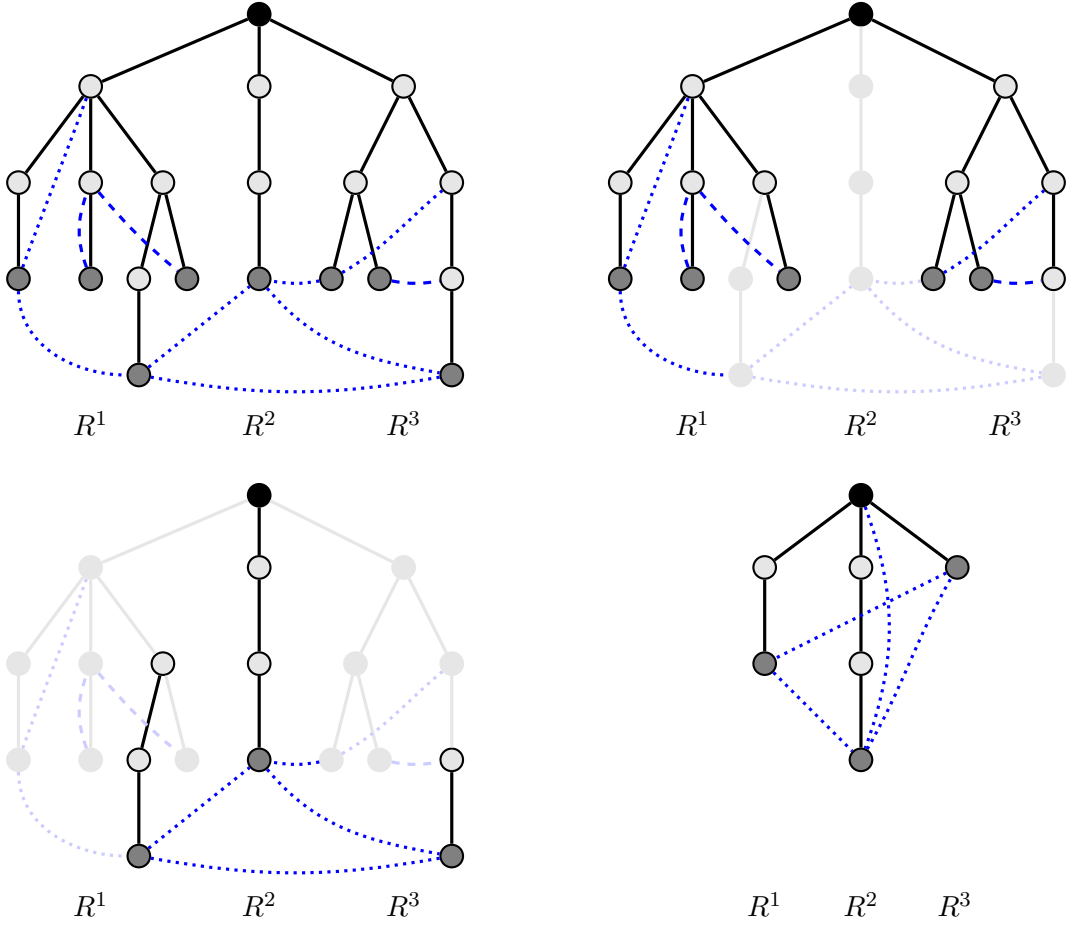
\begin{figure}[h]
\begin{center}
\begin{tikzpicture}[scale=0.32]

\begin{scope}
% \draw [draw=black!40, very thick, dashed] (10,-1) -- (10,17);

\begin{scope}[xshift=-2cm]
\node (a1) at (2,6) [circle,draw=black!100,fill=black!50,thick,inner sep=1pt,minimum size=3mm] {};
\node (a2) at (2,10) [circle,draw=black!100,fill=black!10,thick,inner sep=1pt,minimum size=3mm] {};
\node (a3) at (5,6) [circle,draw=black!100,fill=black!50,thick,inner sep=1pt,minimum size=3mm] {};
\node (a4) at (5,10) [circle,draw=black!100,fill=black!10,thick,inner sep=1pt,minimum size=3mm] {};
\node (a5) at (5,14) [circle,draw=black!100,fill=black!10,thick,inner sep=1pt,minimum size=3mm] {};
\node (a6) at (7,2) [circle,draw=black!100,fill=black!50,thick,inner sep=1pt,minimum size=3mm] {};
\node (a7) at (7,6) [circle,draw=black!100,fill=black!10,thick,inner sep=1pt,minimum size=3mm] {};
\node (a8) at (8,10) [circle,draw=black!100,fill=black!10,thick,inner sep=1pt,minimum size=3mm] {};
\node (a9) at (9,6) [circle,draw=black!100,fill=black!50,thick,inner sep=1pt,minimum size=3mm] {};
\end{scope}
\begin{scope}[xshift=2cm]
\node (a10) at (11,6) [circle,draw=black!100,fill=black!50,thick,inner sep=1pt,minimum size=3mm] {};
\node (a11) at (12,10) [circle,draw=black!100,fill=black!10,thick,inner sep=1pt,minimum size=3mm] {};
\node (a12) at (13,6) [circle,draw=black!100,fill=black!50,thick,inner sep=1pt,minimum size=3mm] {};
\node (a13) at (16,2) [circle,draw=black!100,fill=black!50,thick,inner sep=1pt,minimum size=3mm] {};
\node (a14) at (16,6) [circle,draw=black!100,fill=black!10,thick,inner sep=1pt,minimum size=3mm] {};
\node (a15) at (16,10) [circle,draw=black!100,fill=black!10,thick,inner sep=1pt,minimum size=3mm] {};
\node (a16) at (14,14) [circle,draw=black!100,fill=black!10,thick,inner sep=1pt,minimum size=3mm] {};
\end{scope}

\node (a17) at (10,17) [circle,draw=black!100,fill=black!100,thick,inner sep=1pt,minimum size=3mm] {};

\node (a18) at (10,6) [circle,draw=black!100,fill=black!50,thick,inner sep=1pt,minimum size=3mm] {};
\node (a19) at (10,10) [circle,draw=black!100,fill=black!10,thick,inner sep=1pt,minimum size=3mm] {};
\node (a20) at (10,14) [circle,draw=black!100,fill=black!10,thick,inner sep=1pt,minimum size=3mm] {};

\draw [-,black,very thick] (a1) to node [black,above] {} (a2);
\draw [-,black,very thick] (a3) to node [black,above] {} (a4);
\draw [-,black,very thick] (a2) to node [black,above] {} (a5);
\draw [-,black,very thick] (a4) to node [black,above] {} (a5);
\draw [-,black,very thick] (a6) to node [black,above] {} (a7);
\draw [-,black,very thick] (a7) to node [black,above] {} (a8);
\draw [-,black,very thick] (a9) to node [black,above] {} (a8);
\draw [-,black,very thick] (a8) to node [black,above] {} (a5);
\draw [-,black,very thick] (a5) to node [black,above] {} (a17);
\draw [-,black,very thick] (a10) to node [black,above] {} (a11);
\draw [-,black,very thick] (a12) to node [black,above] {} (a11);
\draw [-,black,very thick] (a13) to node [black,above] {} (a14);
\draw [-,black,very thick] (a14) to node [black,above] {} (a15);
\draw [-,black,very thick] (a15) to node [black,above] {} (a16);
\draw [-,black,very thick] (a11) to node [black,above] {} (a16);
\draw [-,black,very thick] (a16) to node [black,above] {} (a17);

\draw [-,black,very thick] (a17) to node [black,above] {} (a20);
\draw [-,black,very thick] (a20) to node [black,above] {} (a19);
\draw [-,black,very thick] (a19) to node [black,above] {} (a18);

\draw [-,blue,very thick,dotted,out=-90,in=-180] (a1) to node [black] {} (a6);
\draw [-,blue,very thick,dotted,out=70,in=-110] (a1) to node [black] {} (a5);
\draw [-,blue,very thick,dashed,out=110,in=-110] (a3) to node [black] {} (a4);
\draw [-,blue,very thick,dashed,out=-50,in=140] (a4) to node [black] {} (a9);
\draw [-,blue,very thick,dotted] (a6) to node [black] {} (a18);
\draw [-,blue,very thick,dotted,out=-50,in=170] (a18) to node [black] {} (a13);
\draw [-,blue,very thick,dotted,out=-170,in=-10] (a13) to node [black] {} (a6);
\draw [-,blue,very thick,dotted,out=30,in=-140] (a10) to node [black] {} (a15);
\draw [-,blue,very thick,dotted,out=-170,in=-10] (a10) to node [black] {} (a18);
\draw [-,blue,very thick,dashed,out=-10,in=-170] (a12) to node [black] {} (a14);

\node (k1) at (3,0) [circle,draw=black!0,fill=black!0,thick,inner sep=1pt,minimum size=3mm] {$R^1$};
\node (k2) at (10,0) [circle,draw=black!0,fill=black!0,thick,inner sep=1pt,minimum size=3mm] {$R^2$};
\node (k2) at (16,0) [circle,draw=black!0,fill=black!0,thick,inner sep=1pt,minimum size=3mm] {$R^3$};
\end{scope}

%%%%%%%%%%%%%%%%%%%%%%%%%%%%%%%%%%

\begin{scope}[xshift=25cm]
\begin{scope}[xshift=-2cm]
\node (a1) at (2,6) [circle,draw=black!100,fill=black!50,thick,inner sep=1pt,minimum size=3mm] {};
\node (a2) at (2,10) [circle,draw=black!100,fill=black!10,thick,inner sep=1pt,minimum size=3mm] {};
\node (a3) at (5,6) [circle,draw=black!100,fill=black!50,thick,inner sep=1pt,minimum size=3mm] {};
\node (a4) at (5,10) [circle,draw=black!100,fill=black!10,thick,inner sep=1pt,minimum size=3mm] {};
\node (a5) at (5,14) [circle,draw=black!100,fill=black!10,thick,inner sep=1pt,minimum size=3mm] {};
\node (a6) at (7,2) [circle,draw=black!10,fill=black!10,thick,inner sep=1pt,minimum size=3mm] {};
\node (a7) at (7,6) [circle,draw=black!10,fill=black!10,thick,inner sep=1pt,minimum size=3mm] {};
\node (a8) at (8,10) [circle,draw=black!100,fill=black!10,thick,inner sep=1pt,minimum size=3mm] {};
\node (a9) at (9,6) [circle,draw=black!100,fill=black!50,thick,inner sep=1pt,minimum size=3mm] {};
\end{scope}
\begin{scope}[xshift=2cm]
\node (a10) at (11,6) [circle,draw=black!100,fill=black!50,thick,inner sep=1pt,minimum size=3mm] {};
\node (a11) at (12,10) [circle,draw=black!100,fill=black!10,thick,inner sep=1pt,minimum size=3mm] {};
\node (a12) at (13,6) [circle,draw=black!100,fill=black!50,thick,inner sep=1pt,minimum size=3mm] {};
\node (a13) at (16,2) [circle,draw=black!10,fill=black!10,thick,inner sep=1pt,minimum size=3mm] {};
\node (a14) at (16,6) [circle,draw=black!100,fill=black!10,thick,inner sep=1pt,minimum size=3mm] {};
\node (a15) at (16,10) [circle,draw=black!100,fill=black!10,thick,inner sep=1pt,minimum size=3mm] {};
\node (a16) at (14,14) [circle,draw=black!100,fill=black!10,thick,inner sep=1pt,minimum size=3mm] {};
\end{scope}

\node (a17) at (10,17) [circle,draw=black!100,fill=black!100,thick,inner sep=1pt,minimum size=3mm] {};

\node (a18) at (10,6) [circle,draw=black!10,fill=black!10,thick,inner sep=1pt,minimum size=3mm] {};
\node (a19) at (10,10) [circle,draw=black!10,fill=black!10,thick,inner sep=1pt,minimum size=3mm] {};
\node (a20) at (10,14) [circle,draw=black!10,fill=black!10,thick,inner sep=1pt,minimum size=3mm] {};

\draw [-,black,very thick] (a1) to node [black,above] {} (a2);
\draw [-,black,very thick] (a3) to node [black,above] {} (a4);
\draw [-,black,very thick] (a2) to node [black,above] {} (a5);
\draw [-,black,very thick] (a4) to node [black,above] {} (a5);
\draw [-,draw=black!10,very thick] (a6) to node [black,above] {} (a7);
\draw [-,draw=black!10,very thick] (a7) to node [black,above] {} (a8);
\draw [-,black,very thick] (a9) to node [black,above] {} (a8);
\draw [-,black,very thick] (a8) to node [black,above] {} (a5);
\draw [-,black,very thick] (a5) to node [black,above] {} (a17);
\draw [-,black,very thick] (a10) to node [black,above] {} (a11);
\draw [-,black,very thick] (a12) to node [black,above] {} (a11);
\draw [-,draw=black!10,very thick] (a13) to node [black,above] {} (a14);
\draw [-,black,very thick] (a14) to node [black,above] {} (a15);
\draw [-,black,very thick] (a15) to node [black,above] {} (a16);
\draw [-,black,very thick] (a11) to node [black,above] {} (a16);
\draw [-,black,very thick] (a16) to node [black,above] {} (a17);

\draw [-,draw=black!10,very thick] (a17) to node [black,above] {} (a20);
\draw [-,draw=black!10,very thick] (a20) to node [black,above] {} (a19);
\draw [-,draw=black!10,very thick] (a19) to node [black,above] {} (a18);

\draw [-,blue,very thick,dotted,out=-90,in=-180] (a1) to node [black] {} (a6);
\draw [-,blue,very thick,dotted,out=70,in=-110] (a1) to node [black] {} (a5);
\draw [-,blue,very thick,dashed,out=110,in=-110] (a3) to node [black] {} (a4);
\draw [-,blue,very thick,dashed,out=-50,in=140] (a4) to node [black] {} (a9);
\draw [-,draw=blue!20,very thick,dotted] (a6) to node [black] {} (a18);
\draw [-,draw=blue!20,very thick,dotted,out=-50,in=170] (a18) to node [black] {} (a13);
\draw [-,draw=blue!20,very thick,dotted,out=-170,in=-10] (a13) to node [black] {} (a6);
\draw [-,blue,very thick,dotted,out=30,in=-140] (a10) to node [black] {} (a15);
\draw [-,draw=blue!20,very thick,dotted,out=-170,in=-10] (a10) to node [black] {} (a18);
\draw [-,blue,very thick,dashed,out=-10,in=-170] (a12) to node [black] {} (a14);

\node (k1) at (3,0) [circle,draw=black!0,fill=black!0,thick,inner sep=1pt,minimum size=3mm] {$R^1$};
\node (k2) at (10,0) [circle,draw=black!0,fill=black!0,thick,inner sep=1pt,minimum size=3mm] {$R^2$};
\node (k2) at (16,0) [circle,draw=black!0,fill=black!0,thick,inner sep=1pt,minimum size=3mm] {$R^3$};
\end{scope}

%%%%%%%%%%%%%%%%%%%%%%%%%%%%%%%%%% Bottom Left

\begin{scope}[yshift=-20cm]
% \draw [draw=black!40, very thick, dashed] (10,-1) -- (10,17);

\begin{scope}[xshift=-2cm]
\node (a1) at (2,6) [circle,draw=black!10,fill=black!10,thick,inner sep=1pt,minimum size=3mm] {};
\node (a2) at (2,10) [circle,draw=black!10,fill=black!10,thick,inner sep=1pt,minimum size=3mm] {};
\node (a3) at (5,6) [circle,draw=black!10,fill=black!10,thick,inner sep=1pt,minimum size=3mm] {};
\node (a4) at (5,10) [circle,draw=black!10,fill=black!10,thick,inner sep=1pt,minimum size=3mm] {};
\node (a5) at (5,14) [circle,draw=black!10,fill=black!10,thick,inner sep=1pt,minimum size=3mm] {};
\node (a6) at (7,2) [circle,draw=black!100,fill=black!50,thick,inner sep=1pt,minimum size=3mm] {};
\node (a7) at (7,6) [circle,draw=black!100,fill=black!10,thick,inner sep=1pt,minimum size=3mm] {};
\node (a8) at (8,10) [circle,draw=black!100,fill=black!10,thick,inner sep=1pt,minimum size=3mm] {};
\node (a9) at (9,6) [circle,draw=black!10,fill=black!10,thick,inner sep=1pt,minimum size=3mm] {};
\end{scope}
\begin{scope}[xshift=2cm]
\node (a10) at (11,6) [circle,draw=black!10,fill=black!10,thick,inner sep=1pt,minimum size=3mm] {};
\node (a11) at (12,10) [circle,draw=black!10,fill=black!10,thick,inner sep=1pt,minimum size=3mm] {};
\node (a12) at (13,6) [circle,draw=black!10,fill=black!10,thick,inner sep=1pt,minimum size=3mm] {};
\node (a13) at (16,2) [circle,draw=black!100,fill=black!50,thick,inner sep=1pt,minimum size=3mm] {};
\node (a14) at (16,6) [circle,draw=black!100,fill=black!10,thick,inner sep=1pt,minimum size=3mm] {};
\node (a15) at (16,10) [circle,draw=black!10,fill=black!10,thick,inner sep=1pt,minimum size=3mm] {};
\node (a16) at (14,14) [circle,draw=black!10,fill=black!10,thick,inner sep=1pt,minimum size=3mm] {};
\end{scope}

\node (a17) at (10,17) [circle,draw=black!100,fill=black!100,thick,inner sep=1pt,minimum size=3mm] {};

\node (a18) at (10,6) [circle,draw=black!100,fill=black!50,thick,inner sep=1pt,minimum size=3mm] {};
\node (a19) at (10,10) [circle,draw=black!100,fill=black!10,thick,inner sep=1pt,minimum size=3mm] {};
\node (a20) at (10,14) [circle,draw=black!100,fill=black!10,thick,inner sep=1pt,minimum size=3mm] {};

\draw [-,draw=black!10,very thick] (a1) to node [black,above] {} (a2);
\draw [-,draw=black!10,very thick] (a3) to node [black,above] {} (a4);
\draw [-,draw=black!10,very thick] (a2) to node [black,above] {} (a5);
\draw [-,draw=black!10,very thick] (a4) to node [black,above] {} (a5);

\draw [-,black,very thick] (a6) to node [black,above] {} (a7);
\draw [-,black,very thick] (a7) to node [black,above] {} (a8);

\draw [-,draw=black!10,very thick] (a9) to node [black,above] {} (a8);
\draw [-,draw=black!10,very thick] (a8) to node [black,above] {} (a5);
\draw [-,draw=black!10,very thick] (a5) to node [black,above] {} (a17);
\draw [-,draw=black!10,very thick] (a10) to node [black,above] {} (a11);
\draw [-,draw=black!10,very thick] (a12) to node [black,above] {} (a11);
\draw [-,black,very thick] (a13) to node [black,above] {} (a14);
\draw [-,draw=black!10,very thick] (a14) to node [black,above] {} (a15);
\draw [-,draw=black!10,very thick] (a15) to node [black,above] {} (a16);
\draw [-,draw=black!10,very thick] (a11) to node [black,above] {} (a16);
\draw [-,draw=black!10,very thick] (a16) to node [black,above] {} (a17);

\draw [-,black,very thick] (a17) to node [black,above] {} (a20);
\draw [-,black,very thick] (a20) to node [black,above] {} (a19);
\draw [-,black,very thick] (a19) to node [black,above] {} (a18);

% links

\draw [-,draw=blue!20,very thick,dotted,out=-90,in=-180] (a1) to node [black] {} (a6);
\draw [-,draw=blue!20,very thick,dotted,out=70,in=-110] (a1) to node [black] {} (a5);
\draw [-,draw=blue!20,very thick,dashed,out=110,in=-110] (a3) to node [black] {} (a4);
\draw [-,draw=blue!20,very thick,dashed,out=-50,in=140] (a4) to node [black] {} (a9);
\draw [-,blue,very thick,dotted] (a6) to node [black] {} (a18);
\draw [-,blue,very thick,dotted,out=-50,in=170] (a18) to node [black] {} (a13);
\draw [-,blue,very thick,dotted,out=-170,in=-10] (a13) to node [black] {} (a6);
\draw [-,draw=blue!20,very thick,dotted,out=30,in=-140] (a10) to node [black] {} (a15);
\draw [-,blue,very thick,dotted,out=-170,in=-10] (a10) to node [black] {} (a18);
\draw [-,draw=blue!20,very thick,dashed,out=-10,in=-170] (a12) to node [black] {} (a14);

\node (k1) at (3,0) [circle,draw=black!0,fill=black!0,thick,inner sep=1pt,minimum size=3mm] {$R^1$};
\node (k2) at (10,0) [circle,draw=black!0,fill=black!0,thick,inner sep=1pt,minimum size=3mm] {$R^2$};
\node (k2) at (16,0) [circle,draw=black!0,fill=black!0,thick,inner sep=1pt,minimum size=3mm] {$R^3$};
\end{scope}

%%%%%%%%%%%%%%%%%%%%%%%%%%%%%%%%%%%%%%%% Bottom Right

\begin{scope}[yshift=-20cm, xshift=25cm]
% \draw [draw=black!40, very thick, dashed] (10,-1) -- (10,17);

\begin{scope}[xshift=-2cm]
% \node (a1) at (2,6) [circle,draw=black!10,fill=black!10,thick,inner sep=1pt,minimum size=3mm] {};
% \node (a2) at (2,10) [circle,draw=black!10,fill=black!10,thick,inner sep=1pt,minimum size=3mm] {};
% \node (a3) at (5,6) [circle,draw=black!10,fill=black!10,thick,inner sep=1pt,minimum size=3mm] {};
% \node (a4) at (5,10) [circle,draw=black!10,fill=black!10,thick,inner sep=1pt,minimum size=3mm] {};
% \node (a5) at (5,14) [circle,draw=black!10,fill=black!10,thick,inner sep=1pt,minimum size=3mm] {};
\node (a6) at (8,10) [circle,draw=black!100,fill=black!50,thick,inner sep=1pt,minimum size=3mm] {};
\node (a7) at (8,14) [circle,draw=black!100,fill=black!10,thick,inner sep=1pt,minimum size=3mm] {};
% \node (a8) at (8,14) [circle,draw=black!100,fill=black!10,thick,inner sep=1pt,minimum size=3mm] {};
% \node (a9) at (9,6) [circle,draw=black!10,fill=black!10,thick,inner sep=1pt,minimum size=3mm] {};
\end{scope}
\begin{scope}[xshift=2cm]
% \node (a10) at (11,6) [circle,draw=black!10,fill=black!10,thick,inner sep=1pt,minimum size=3mm] {};
% \node (a11) at (12,10) [circle,draw=black!10,fill=black!10,thick,inner sep=1pt,minimum size=3mm] {};
% \node (a12) at (13,6) [circle,draw=black!10,fill=black!10,thick,inner sep=1pt,minimum size=3mm] {};
\node (a13) at (12,14) [circle,draw=black!100,fill=black!50,thick,inner sep=1pt,minimum size=3mm] {};
% \node (a14) at (12,14) [circle,draw=black!100,fill=black!10,thick,inner sep=1pt,minimum size=3mm] {};
% \node (a15) at (16,10) [circle,draw=black!10,fill=black!10,thick,inner sep=1pt,minimum size=3mm] {};
% \node (a16) at (14,14) [circle,draw=black!10,fill=black!10,thick,inner sep=1pt,minimum size=3mm] {};
\end{scope}

\node (a17) at (10,17) [circle,draw=black!100,fill=black!100,thick,inner sep=1pt,minimum size=3mm] {};

\node (a18) at (10,6) [circle,draw=black!100,fill=black!50,thick,inner sep=1pt,minimum size=3mm] {};
\node (a19) at (10,10) [circle,draw=black!100,fill=black!10,thick,inner sep=1pt,minimum size=3mm] {};
\node (a20) at (10,14) [circle,draw=black!100,fill=black!10,thick,inner sep=1pt,minimum size=3mm] {};

% \draw [-,draw=black!10,very thick] (a1) to node [black,above] {} (a2);
% \draw [-,draw=black!10,very thick] (a3) to node [black,above] {} (a4);
% \draw [-,draw=black!10,very thick] (a2) to node [black,above] {} (a5);
% \draw [-,draw=black!10,very thick] (a4) to node [black,above] {} (a5);

\draw [-,black,very thick] (a6) to node [black,above] {} (a7);
\draw [-,black,very thick] (a7) to node [black,above] {} (a17);
% \draw [-,black,very thick] (a8) to node [black,above] {} (a17);

% \draw [-,draw=black!10,very thick] (a9) to node [black,above] {} (a8);
% \draw [-,draw=black!10,very thick] (a8) to node [black,above] {} (a5);
% \draw [-,draw=black!10,very thick] (a5) to node [black,above] {} (a17);
% \draw [-,draw=black!10,very thick] (a10) to node [black,above] {} (a11);
% \draw [-,draw=black!10,very thick] (a12) to node [black,above] {} (a11);
\draw [-,black,very thick] (a13) to node [black,above] {} (a17);
% \draw [-,black,very thick] (a14) to node [black,above] {} (a17);
% \draw [-,draw=black!10,very thick] (a14) to node [black,above] {} (a15);
% \draw [-,draw=black!10,very thick] (a15) to node [black,above] {} (a16);
% \draw [-,draw=black!10,very thick] (a11) to node [black,above] {} (a16);
% \draw [-,draw=black!10,very thick] (a16) to node [black,above] {} (a17);

\draw [-,black,very thick] (a17) to node [black,above] {} (a20);
\draw [-,black,very thick] (a20) to node [black,above] {} (a19);
\draw [-,black,very thick] (a19) to node [black,above] {} (a18);

% links

% \draw [-,draw=blue!20,very thick,dotted,out=-90,in=-180] (a1) to node [black] {} (a6);
% \draw [-,draw=blue!20,very thick,dotted,out=70,in=-110] (a1) to node [black] {} (a5);
% \draw [-,draw=blue!20,very thick,dashed,out=110,in=-110] (a3) to node [black] {} (a4);
% \draw [-,draw=blue!20,very thick,dashed,out=-50,in=140] (a4) to node [black] {} (a9);
\draw [-,blue,very thick,dotted] (a6) to node [black] {} (a18);
\draw [-,blue,very thick,dotted,out=60,in=-120] (a18) to node [black] {} (a13);
\draw [-,blue,very thick,dotted] (a13) to node [black] {} (a6);
% \draw [-,draw=blue!20,very thick,dotted,out=30,in=-140] (a10) to node [black] {} (a15);
\draw [-,blue,very thick,dotted,out=-70,in=70] (a17) to node [black] {} (a18);
% \draw [-,draw=blue!20,very thick,dashed,out=-10,in=-170] (a12) to node [black] {} (a14);

\node (k1) at (6,0) [circle,draw=black!0,fill=black!0,thick,inner sep=1pt,minimum size=3mm] {$R^1$};
\node (k2) at (10,0) [circle,draw=black!0,fill=black!0,thick,inner sep=1pt,minimum size=3mm] {$R^2$};
\node (k2) at (13.9,0) [circle,draw=black!0,fill=black!0,thick,inner sep=1pt,minimum size=3mm] {$R^3$};
\end{scope}

\end{tikzpicture}
\end{center}
\caption{Illustration of the proof of Lemma~\ref{lem:cross-rounding}. 
Top Left: A tree $T$ and the support of $z$ are shown. Dashes and dotted links represent 
integral and half-integral links, respectively. Top Right: The set $E_\lambda$ is shown, for $\lambda=2$,
as well as all in-links partially covering $E_{\lambda}$. Bottom Left: The edges of the tree $T'$ and
the cross-links partially covering $E[T']$. Bottom Right: The tree $T'$ and the support of the 
solution $y'$, obtained by contracting $E_\lambda$. The instance is star-shaped.}
\label{fig:star-shaped-rounding}
\end{figure}

We apply Lemma~\ref{lem:cross-rounding} to round $z$ with $\lambda = 3 + \sqrt{5}$, 
whenever $\frac{c^\intercal z^{cr}}{c^\intercal z} \geq \alpha^*$, for some 
$\alpha^* \in (0,1)$ that we determine later. The values of $\lambda$ and $\alpha^*$
are simply chosen to minimize the overall approximation guarantee. 

We proceed to the second rounding procedure, that is used to round $z$ 
when the ratio $\frac{c^\intercal z^{cr}}{c^\intercal z}$ is small, namely
when $\frac{c^\intercal z^{cr}}{c^\intercal z} < \alpha^*$.

\subsubsection{Second Rounding Procedure: Nearly Decomposable Pairs}

Our second rounding procedure proceeds in two steps. First, 
each cross-link $\ell = uv \in L$ is replaced in the fractional
solution $z$ with the two shadows $\ell_u = ur$ and $\ell_v = vr$,
by adding $z_\ell$ to the $\ell_u$-th and $\ell_v$-th components of $z$, 
and setting $z_\ell$ to zero.
This way, we obtain a solution that has no cross-links in its support.
Formally, create a new solution $y\in \mathbb{R}_{\geq 0}^L$ derived
from $z$ by setting for each $\ell = pq \in L$
$$
y_\ell = \begin{cases} 
				     z_\ell & \mbox{if } \,\,\,  \ell\in \supp{z^{in}}, \,\, r\not\in \ell \\ 
				     z_\ell + \sum_{\ell' \in \supp{z^{cr}},\,\, q\in \ell'} z_{\ell'} & \mbox{if } \,\,\, p=r, \, q\in V[T]\setminus \{r\} \\
				     0 & \mbox{otherwise.} \\
\end{cases}
$$
Clearly, $c^\intercal y = c^\intercal z^{in} + 2c^\intercal z^{cr}$ holds, as the costs 
of any link is at least the cost of any of its shadows. Next, split the tree $T$ into 
subtrees $\bar R^1, \cdots, \bar R^m$, by separating $T$ at the root $r$ (note that
$R^j \neq \bar R^j$, as each $\bar R_j$ also contains the root $r$). Now, since $y$ contains
no cross-links in its support, it is a union of $m$ disjoint solutions, one for each 
subtree $\bar R^j$. The latter partition of $(T,z)$ into $m$ parts concludes the first step.

We can henceforth focus on a single part corresponding to a subtree 
$\bar R\in \{\bar R^1,\cdots, \bar R^m\}$, and present a rounding 
procedure for its corresponding part of $y$, namely for $\bar y$, defined as
$$
\bar y_\ell = \begin{cases} 
				     y_\ell & \mbox{if } \,\,\,  \ell\in V[\bar R] \times V[\bar R] \\ 
				     0 & \mbox{otherwise,} \\
\end{cases}
$$
for $\ell \in L$. The union of the obtained solutions comprises a WTAP solution for $T$. 
In the following lemma we show how to exploit the bundle constraints to prove that
$\bar y$ can be rounded up with practically no loss in terms of cost. We charge some of the
cost to the early compound nodes and use the fact that $(T,z)$ is $\beta$-simple for 
$\beta = \frac{48M}{\epsilon^2}$.
% We denote by $\OPT(\bar R)$ the optimal solution to the WTAP instance restricted to the subtree 
% $\bar R$, and to the links in $V[\bar R] \times V[\bar R]$.

\begin{lemma}\label{lem:bundle-rounding-property}
 If $x$ is an optimal solution to $\mathrm{LP}_\gamma$ for $\gamma \geq \frac{200M}{\epsilon^2}$, then
$$
\opt{E[\bar R]} \leq c^\intercal \bar y + \sum_{u\in V[\bar R]} s_u.
$$
\end{lemma}

\begin{proof}
%Recall that $T$ is $\beta$-simple for $\beta = \frac{48M}{\epsilon^2}$ and that $T$ satisfies the
%thin coverage property, thus
$(T,z)$ is $\beta$-simple and the thin coverage property is satisfied thus
$
c^\intercal \bar y \leq \frac{48M}{\epsilon^2} + \frac{2M}{\epsilon} \leq \frac{50M}{\epsilon^2}.
$
It follows, from feasibility of $\bar y$ and Proposition~\ref{prop:simple_rounding} that 
$\opt{E[\bar R]} \leq \frac{100M}{\epsilon^2}$. 
Now, if $\bar R$ contains at least $\frac{100M}{\epsilon^2}$ early compound nodes then 
$\sum_{u\in V[\bar R]} s_u \geq \frac{100M}{\epsilon^2} \geq \opt{E[\bar R]}$, and we are done.

In the other case the number of early compound nodes in $\bar R$ is at most $\frac{100M}{\epsilon^2}$.
Our goal is to prove that in this case $E[\bar R]$ is a $\frac{200M}{\epsilon^2}$-bundle in $G$, 
and hence if $\gamma \geq \frac{200M}{\epsilon^2}$, $\mathrm{LP}_{\gamma}$ contains the  constraint
$$
\sum_{\ell \in \cov{E[\bar R]}} c_\ell x_\ell \geq \opt{E[\bar R]}.
$$
Since $(T,z)$ is $\beta$-simple, the number of leaves in $\bar R$ is at most 
$\frac{48M}{\epsilon^2} + 1 \leq \frac{50M}{\epsilon^2}$. Let $W \subseteq V[\bar R]$
be the set of nodes with degree at least three in $\bar R$. Since the
number of leaves is at most $\frac{50M}{\epsilon^2}$, also 
$|W|\leq \frac{50M}{\epsilon^2}$ holds. Let 
$Q_1, \cdots, Q_t$ be the $t\leq 2\cdot \frac{50M}{\epsilon^2} = \frac{100M}{\epsilon^2}$
paths obtained by splitting $\bar R$ at the nodes in $W$ (see Figure~\ref{fig:bundle-rounding}).

%so its edge set $E[\bar R]$ is a 
%union of at most $\frac{100M}{\epsilon^2}$ paths in $T$, namely it is a $\frac{100M}{\epsilon^2}$-bundle
%in $T$, and hence also in $\bar G$. Denote by $Q_1, \cdots, Q_t$ the $t\leq \frac{100M}{\epsilon^2}$ paths, 
%the unions of which comprises $\bar R$. We assume for simplicity
%that no path $Q_i$ for $i\in [t]$ has a node in its interior, whose degree in $\bar R$ is
%larger than two (this partition can be achieved by separating the paths at nodes of $\bar R$ 
%with degree larger than two).

%We claim that $\bar R$ is a $\frac{200M}{\epsilon^2}$-bundle in the original tree $G$. Indeed,
Now, each path $Q_i$ is a union of paths of $G$, separated by components contracted 
in the first phase of the algorithm, when we contracted the edges in $E^h$, that 
were covered by $L_0$. More precisely, if $Q_i$ is a union of $b\in \mathbb{Z}_{\geq 1}$ paths in $G$, then these
paths are subpaths of $Q_i$, separated by $b-1$ early compound nodes. Furthermore,
since we designed $Q_1, \cdots, Q_t$ to not have nodes of degree larger than two 
in the interior, each early compound node can lie in the interior of at most one path 
of the decomposition. 
Since the total number of early compound nodes is at most $\frac{100M}{\epsilon^2}$,
$\bar R$ is indeed a union of at most $\frac{200M}{\epsilon^2}$ paths in $G$, 
implying that $E[\bar R]$ is a $\frac{200M}{\epsilon^2}$-bundle of $G$, as required.
Figure~\ref{fig:bundle-rounding} illustrates this argument.

%We have proved that if there are at most $\frac{100M}{\epsilon^2}$ early compound
%nodes in $V[\bar R]$, then $E[\bar R]$ is a $\frac{200M}{\epsilon^2}$-bundle. 
%It follows that the $\mathrm{LP}_\gamma$, which is used to obtain $x$ contains the constraint
%$$
%\sum_{\ell \in \cov{E[\bar R]}} c_\ell x_\ell \geq \opt{E[\bar R]}.
%$$
Consequently, it follows from Lemma~\ref{lem:weight-presenving} and the fact that $E[\bar R]\subseteq E[T]$ that
$$
c^\intercal \bar y = 
\sum_{\ell \in \cov{E[\bar R]}} c_\ell \bar y_\ell \geq \sum_{\ell \in \cov{E[\bar R]}} c_\ell z_\ell \geq
\sum_{\ell \in \cov{E[\bar R]}} c_\ell x_\ell \geq \opt{E[\bar R]},
$$
which concludes the proof of the lemma.

\end{proof}

\begin{figure}[h]
\begin{center}
\begin{tikzpicture}[scale=0.35]

\node (a1) at (3,0) [circle,draw=black!10,fill=black!10,thick,inner sep=1pt,minimum size=3mm] {};
\node (a2) at (5,3) [circle,draw=black!10,fill=black!10,thick,inner sep=1pt,minimum size=3mm] {};
\node (a3) at (2,6) [circle,draw=black!100,fill=black!10,thick,inner sep=1pt,minimum size=3mm] {};
\node (a4) at (3,9) [circle,draw=black!100,fill=black!10,thick,inner sep=1pt,minimum size=3mm] {};
\node (a5) at (0,11) [circle,draw=black!100,fill=black!10,thick,inner sep=1pt,minimum size=3mm] {};
\node (a6) at (7,0) [circle,draw=black!10,fill=black!10,thick,inner sep=1pt,minimum size=3mm] {};
\node (a7) at (9,3) [circle,draw=black!100,fill=black!10,thick,inner sep=1pt,minimum size=3mm] {};
\node (a8) at (7,6) [circle,draw=black!10,fill=black!10,thick,inner sep=1pt,minimum size=3mm] {};
\node (a9) at (6,9) [circle,draw=black!10,fill=black!10,thick,inner sep=1pt,minimum size=3mm] {};
\node (a10) at (9,11) [circle,draw=black!10,fill=black!10,thick,inner sep=1pt,minimum size=3mm] {};
\node (a11) at (10,0) [circle,draw=black!100,fill=black!10,thick,inner sep=1pt,minimum size=3mm] {};
\node (a12) at (12,3) [circle,draw=black!10,fill=black!10,thick,inner sep=1pt,minimum size=3mm] {};
\node (a13) at (14,0) [circle,draw=black!10,fill=black!10,thick,inner sep=1pt,minimum size=3mm] {};
\node (a14) at (12,6) [circle,draw=black!10,fill=black!10,thick,inner sep=1pt,minimum size=3mm] {};
\node (a15) at (12,9) [circle,draw=black!100,fill=black!10,thick,inner sep=1pt,minimum size=3mm] {};
\node (a16) at (12,12) [circle,draw=black!100,fill=black!10,thick,inner sep=1pt,minimum size=3mm] {};
\node (a17) at (18,2) [circle,draw=black!10,fill=black!10,thick,inner sep=1pt,minimum size=3mm] {};
\node (a18) at (15,4) [circle,draw=black!10,fill=black!10,thick,inner sep=1pt,minimum size=3mm] {};
\node (a19) at (18,6) [circle,draw=black!10,fill=black!10,thick,inner sep=1pt,minimum size=3mm] {};
\node (a20) at (15,8) [circle,draw=black!10,fill=black!10,thick,inner sep=1pt,minimum size=3mm] {};
\node (a21) at (18,10) [circle,draw=black!10,fill=black!10,thick,inner sep=1pt,minimum size=3mm] {};
\node (a22) at (14,11) [circle,draw=black!100,fill=black!10,thick,inner sep=1pt,minimum size=3mm] {};

\draw [-,black,very thick,dashed] (a1) to node [black] {} (a2);
\draw [-,black,very thick,dashed] (a2) to node [black] {} (a6);
\draw [-,black,very thick,dashed] (a2) to node [black] {} (a8);
\draw [-,black,very thick] (a3) to node [black] {} (a4);
\draw [-,black,very thick] (a4) to node [black] {} (a5);
\draw [-,black,very thick] (a4) to node [black] {} (a9);
\draw [-,black,very thick] (a7) to node [black] {} (a8);
\draw [-,black,very thick,dashed] (a8) to node [black] {} (a9);
\draw [-,black,very thick,dashed] (a9) to node [black] {} (a10);
\draw [-,black,very thick] (a7) to node [black] {} (a11);
\draw [-,black,very thick] (a7) to node [black] {} (a12);
\draw [-,black,very thick,dashed] (a12) to node [black] {} (a13);
\draw [-,black,very thick,dashed] (a12) to node [black] {} (a14);
\draw [-,black,very thick] (a14) to node [black] {} (a15);
\draw [-,black,very thick] (a15) to node [black] {} (a16);
\draw [-,black,very thick] (a15) to node [black] {} (a22);
\draw [-,black,very thick,dashed] (a12) to node [black] {} (a18);
\draw [-,black,very thick,dashed] (a14) to node [black] {} (a20);
\draw [-,black,very thick,dashed] (a17) to node [black] {} (a18);
\draw [-,black,very thick,dashed] (a19) to node [black] {} (a20);
\draw [-,black,very thick,dashed] (a20) to node [black] {} (a21);

\draw [-,blue,very thick, out=80,dotted, in=-80] (a17) to node [black] {} (a19);
\draw [-,blue,very thick, out=60,dotted, in=-60] (a17) to node [black] {} (a21);

\draw [-,blue,very thick, out=80,dotted, in=-120] (a1) to node [black] {} (a10);
\draw [-,blue,very thick, out=70,dotted, in=-90] (a6) to node [black] {} (a10);

\draw [-,blue,very thick, out=10,dotted, in=-130] (a13) to node [black] {} (a17);

\begin{scope}[xshift=25cm]

\node (a3) at (2,6) [circle,draw=black!100,fill=black!10,thick,inner sep=1pt,minimum size=3mm] {};
\node (a4) at (3,9) [circle,draw=black!100,fill=black!10,thick,inner sep=1pt,minimum size=3mm] {};
\node (a5) at (0,11) [circle,draw=black!100,fill=black!10,thick,inner sep=1pt,minimum size=3mm] {};
\node (a7) at (9,3) [circle,draw=black!100,fill=black!10,thick,inner sep=1pt,minimum size=3mm] {};
\node (a11) at (10,0) [circle,draw=black!100,fill=black!10,thick,inner sep=1pt,minimum size=3mm] {};
\node (c1) at (7,7) [circle,draw=black!100,fill=black!100,thick,inner sep=1pt,minimum size=4mm] {};
\node (a15) at (12,9) [circle,draw=black!100,fill=black!10,thick,inner sep=1pt,minimum size=3mm] {};
\node (a16) at (12,12) [circle,draw=black!100,fill=black!10,thick,inner sep=1pt,minimum size=3mm] {};
\node (c2) at (11,6) [circle,draw=black!100,fill=black!100,thick,inner sep=1pt,minimum size=4mm] {};
\node (a22) at (14,11) [circle,draw=black!100,fill=black!10,thick,inner sep=1pt,minimum size=3mm] {};

\draw [-,black,very thick] (a3) to node [black] {} (a4);
\draw [-,black,very thick] (a4) to node [black] {} (a5);
\draw [-,black,very thick] (a7) to node [black] {} (a11);
\draw [-,black,very thick] (a15) to node [black] {} (a16);
\draw [-,black,very thick] (a15) to node [black] {} (a22);

\draw [-,black,very thick] (a4) to node [black] {} (c1);
\draw [-,black,very thick] (a7) to node [black] {} (c1);

\draw [-,black,very thick] (a7) to node [black] {} (c2);
\draw [-,black,very thick] (a15) to node [black] {} (c2);

\end{scope}

\begin{scope}[xshift=0cm, yshift=-16cm]

\node (a3) at (2,6) [circle,draw=black!100,fill=black!10,thick,inner sep=1pt,minimum size=3mm] {};

\node (a4a3) at (2.66,8) [circle,draw=black!100,fill=black!10,thick,inner sep=1pt,minimum size=3mm] {};
\node (a4a5) at (2.66,9.25) [circle,draw=black!100,fill=black!10,thick,inner sep=1pt,minimum size=3mm] {};
\node (a4c1) at (3.6,8.6) [circle,draw=black!100,fill=black!10,thick,inner sep=1pt,minimum size=3mm] {};

\node (a5) at (0,11) [circle,draw=black!100,fill=black!10,thick,inner sep=1pt,minimum size=3mm] {};

\node (a7a11) at (9.1,2.66) [circle,draw=black!100,fill=black!10,thick,inner sep=1pt,minimum size=3mm] {};
\node (a7c1) at (8.66,4) [circle,draw=black!100,fill=black!10,thick,inner sep=1pt,minimum size=3mm] {};
\node (a7c2) at (9.75,3.66) [circle,draw=black!100,fill=black!10,thick,inner sep=1pt,minimum size=3mm] {};

\node (a11) at (10,0) [circle,draw=black!100,fill=black!10,thick,inner sep=1pt,minimum size=3mm] {};
\node (c1) at (7,7) [circle,draw=black!100,fill=black!100,thick,inner sep=1pt,minimum size=4mm] {};

\node (a15a16) at (12,9.66) [circle,draw=black!100,fill=black!10,thick,inner sep=1pt,minimum size=3mm] {};
\node (a15a22) at (12.85,9) [circle,draw=black!100,fill=black!10,thick,inner sep=1pt,minimum size=3mm] {};
\node (a15c2) at (11.75,8.5) [circle,draw=black!100,fill=black!10,thick,inner sep=1pt,minimum size=3mm] {};

\node (a16) at (12,12) [circle,draw=black!100,fill=black!10,thick,inner sep=1pt,minimum size=3mm] {};
\node (c2) at (11,6) [circle,draw=black!100,fill=black!100,thick,inner sep=1pt,minimum size=4mm] {};
\node (a22) at (14,11) [circle,draw=black!100,fill=black!10,thick,inner sep=1pt,minimum size=3mm] {};

\draw [-,black,very thick] (a3) to node [black] {} (a4a3);
\draw [-,black,very thick] (a4a5) to node [black] {} (a5);
\draw [-,black,very thick] (a7a11) to node [black] {} (a11);
\draw [-,black,very thick] (a15a16) to node [black] {} (a16);
\draw [-,black,very thick] (a15a22) to node [black] {} (a22);

\draw [-,black,very thick] (a4c1) to node [black] {} (c1);
\draw [-,black,very thick] (a7c1) to node [black] {} (c1);

\draw [-,black,very thick] (a7c2) to node [black] {} (c2);
\draw [-,black,very thick] (a15c2) to node [black] {} (c2);

\end{scope}

\begin{scope}[xshift=25cm, yshift=-16cm]

\node (a3) at (2,6) [circle,draw=black!100,fill=black!10,thick,inner sep=1pt,minimum size=3mm] {};

\node (a4a3) at (2.66,8) [circle,draw=black!100,fill=black!10,thick,inner sep=1pt,minimum size=3mm] {};
\node (a4a5) at (2.66,9.25) [circle,draw=black!100,fill=black!10,thick,inner sep=1pt,minimum size=3mm] {};
\node (a4c1) at (3.6,8.6) [circle,draw=black!100,fill=black!10,thick,inner sep=1pt,minimum size=3mm] {};

\node (a5) at (0,11) [circle,draw=black!100,fill=black!10,thick,inner sep=1pt,minimum size=3mm] {};

\node (a7a11) at (9.1,2.66) [circle,draw=black!100,fill=black!10,thick,inner sep=1pt,minimum size=3mm] {};
\node (a7c1) at (8.66,4) [circle,draw=black!100,fill=black!10,thick,inner sep=1pt,minimum size=3mm] {};
\node (a7c2) at (9.75,3.66) [circle,draw=black!100,fill=black!10,thick,inner sep=1pt,minimum size=3mm] {};

\node (a11) at (10,0) [circle,draw=black!100,fill=black!10,thick,inner sep=1pt,minimum size=3mm] {};

\node (c1UP) at (6.5,7.33) [circle,draw=black!100,fill=black!100,thick,inner sep=1pt,minimum size=3mm] {};
\node (c1DOWN) at (7.33,6.33) [circle,draw=black!100,fill=black!100,thick,inner sep=1pt,minimum size=3mm] {};

\node (a15a16) at (12,9.66) [circle,draw=black!100,fill=black!10,thick,inner sep=1pt,minimum size=3mm] {};
\node (a15a22) at (12.85,9) [circle,draw=black!100,fill=black!10,thick,inner sep=1pt,minimum size=3mm] {};
\node (a15c2) at (11.75,8.5) [circle,draw=black!100,fill=black!10,thick,inner sep=1pt,minimum size=3mm] {};

\node (a16) at (12,12) [circle,draw=black!100,fill=black!10,thick,inner sep=1pt,minimum size=3mm] {};

\node (c2UP) at (11.2,6.5) [circle,draw=black!100,fill=black!100,thick,inner sep=1pt,minimum size=3mm] {};
\node (c2DOWN) at (10.8,5.33) [circle,draw=black!100,fill=black!100,thick,inner sep=1pt,minimum size=3mm] {};

\node (a22) at (14,11) [circle,draw=black!100,fill=black!10,thick,inner sep=1pt,minimum size=3mm] {};

\draw [-,black,very thick] (a3) to node [black] {} (a4a3);
\draw [-,black,very thick] (a4a5) to node [black] {} (a5);
\draw [-,black,very thick] (a7a11) to node [black] {} (a11);
\draw [-,black,very thick] (a15a16) to node [black] {} (a16);
\draw [-,black,very thick] (a15a22) to node [black] {} (a22);

\draw [-,black,very thick] (a4c1) to node [black] {} (c1UP);
\draw [-,black,very thick] (a7c1) to node [black] {} (c1DOWN);

\draw [-,black,very thick] (a7c2) to node [black] {} (c2DOWN);
\draw [-,black,very thick] (a15c2) to node [black] {} (c2UP);
\end{scope}

\end{tikzpicture}
\end{center}
\caption{The proof of Lemma~\ref{lem:bundle-rounding-property}. Top Left: The part of the 
original tree $G$ corresponding to $\bar R$. The dashed edges are contracted, as they are 
part of $E^h$, resulting in two early compound nodes. The dotted links are included 
in $L_0$. Top Right: The subtree $\bar R$. The large full nodes are the early compound
nodes. Bottom Left: Splitting $\bar R$ into $7$ paths by separating at the nodes in $W$.
Bottom Right: The resulting paths are further subdivided at early compound nodes to 
obtain $9$ paths, forming a $9$-bundle in the original tree $G$.}\label{fig:bundle-rounding}
\end{figure}
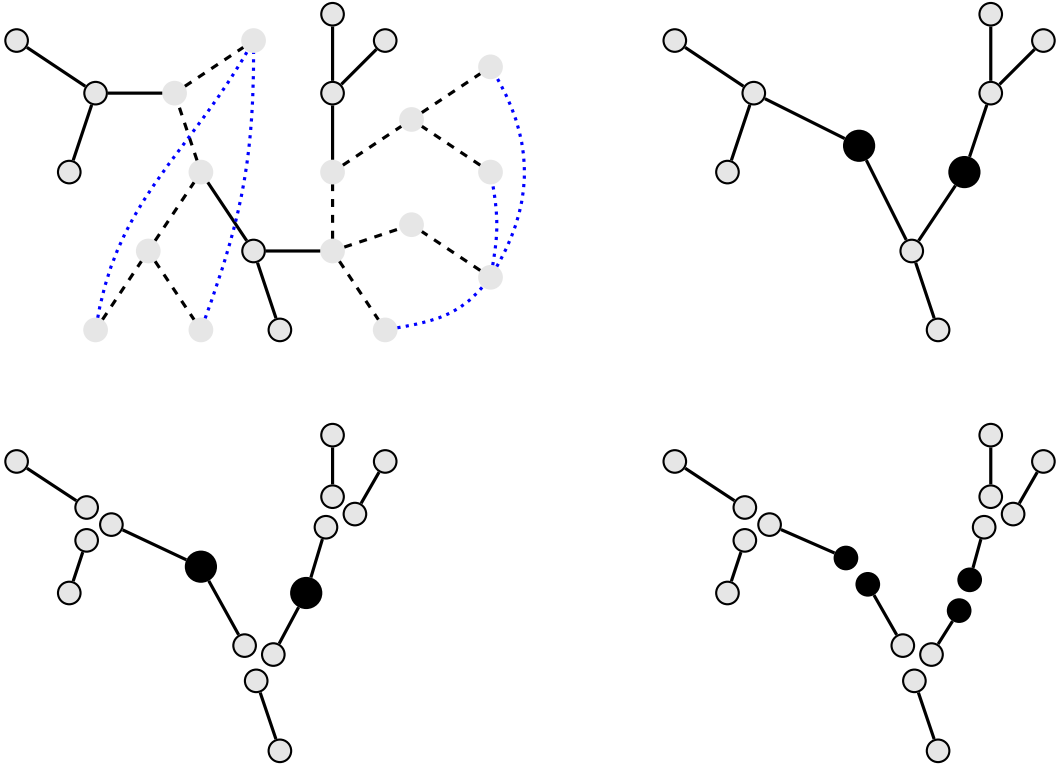

We have proved that the optimal cost of covering $\bar R$ is at most $c^\intercal \bar y$, up to
a term that depends on the cost incurred by contracting links of $L_0$ into early compound
nodes that are contained in $V[\bar R]$. Now, since $\opt{\bar R} \leq \frac{100M}{\epsilon^2} = O(M)$,
we can find an optimal solution by simple enumeration in time $n^{O(M)}$ and report it
as a set of links covering $\bar R$ (in fact, these values were already computed in the construction
of $\mathrm{LP}_\gamma$).
Repeating this for all $\bar R\in \{\bar R^1,\cdots, \bar R^m\}$
concludes the rounding procedure for the pair $(T,z)$. We summarize in the following lemma, 
which is a direct consequence of Lemma~\ref{lem:bundle-rounding-property} and the previous argument.

\begin{lemma}\label{lem:bundle-rounding}
There is an algorithm that given $T$ and $z$, computes in time $n^{O(M)}$ a
set of links $S\subseteq L$ that covers $E[T]$ with cost at most
$$
c(S) \leq c^\intercal z^{in} + 2c^\intercal z^{cr} + \sum_{u\in V[T]} s_u.
$$
\end{lemma}

\subsection{The Overall Algorithm}

The algorithm starts by computing an optimal solution $x$ of 
$\mathrm{LP}_\gamma$ with $\gamma = \frac{200M}{\epsilon^2}$. For
constant $M$, this LP has polynomial size, and each right hand side
of a $\gamma$-bundle constraint can be computed using Lemma~\ref{lem:few-leaves}. 
Then, the first phase of the rounding
algorithm computes the set of links $L_0 \cup L_1$ and
a decomposition of $(\bar G, x)$ into $\beta$-simple pairs 
$(T^1, z^1), \cdots, (T^k, z^k)$ for $\beta = \frac{48M}{\epsilon^2}$.
This concludes the first phase.

In the second phase, each pair $(T^j, z^j)$ is rounded using either
Lemma~\ref{lem:cross-rounding}, or Lemma~\ref{lem:bundle-rounding}. Concretely,
for each $j\in [k]$ the algorithm computes the ratio
$
\alpha^j = \frac{c^\intercal {z^j}^{cr}}{c^\intercal z^j} 
$
and applies Lemma~\ref{lem:cross-rounding} with $\lambda = 3 + \sqrt{5}$ if 
$\alpha^j \geq \alpha^*$, and Lemma~\ref{lem:bundle-rounding}, if $\alpha^j < \alpha^*$,
for
$$
\alpha^* = 1 - \frac{2(5-2\sqrt{5})}{25+2\sqrt{5}} \approx 0.96418,
$$
which is the value for which Lemmas~\ref{lem:cross-rounding} and~\ref{lem:bundle-rounding}
give precisely the same approximation, when $\lambda = 3 + \sqrt{5}$.
We are ready to prove Theorem~\ref{thm:wtap}.

\begin{proof}[Proof of Theorem~\ref{thm:wtap}]
Let $\mathrm{OPT}$ be the value of the optimal solution.
The solution computed by the algorithm is
clearly feasible, and the algorithm runs in polynomial time. It remains to prove the approximation
guarantee.

The cost of the links in $L_0 \cup L_1$ is at most
$2\epsilon c^\intercal x = O(\epsilon) \cdot \mathrm{OPT}$. 
Next, for each $(T,z)$ in the decomposition $(T^1,z^1),\cdots, (T^k,z^k)$, the cost 
incurred for the links covering $T$ is at most
$$
\min \left\{ 2\lambda c^\intercal z^{in} + \frac{4\lambda}{3(\lambda-1)} c^\intercal z^{cr}, \,\,
c^\intercal z^{in} + 2c^\intercal z^{cr} + \sum_{u\in V[T^j]} s_u
\right\} \leq \delta c^\intercal z + \sum_{u\in V[T^j]} s_u,
$$
due to Lemmas~\ref{lem:cross-rounding} and~\ref{lem:bundle-rounding}. The inequality follows
by simple verification, using the definition of the constants $\lambda, \alpha^*$ and $\delta$ and
the fact that $z = z^{in} + z^{cr}$.
Hence, the total cost of links included in the solution in rounding all trees 
in the decomposition is at most
$$
\delta \cdot \sum_{i\in [k]} c^\intercal z^i + \sum_{u\in V[\bar G]} s_u \leq 
(1+O(\epsilon))\delta c^\intercal x \leq (1+O(\epsilon))\delta \cdot \mathrm{OPT}, 
$$
due to Lemma~\ref{lem:weight-increase-decomp} and since 
$\sum_{u\in V[\bar G]} s_u = c(L_0) \leq \epsilon \cdot \mathrm{OPT}$. This concludes the proof.
\end{proof}

\paragraph{Acknowledgments} The author is grateful to Steve Chestnut and 
Akaki Mamageishvili for many valuable comments.

\small
 
% \bibliographystyle{plain}
% \bibliography{lit}

\appendix

\section{Solving the Bundle LP}\label{apx:solve-lp}

Since the number of $\gamma$-bundles is $n^{O(\gamma)}$, it follows
that $\mathrm{LP}_\gamma$ can be solved in polynomial time, if for a
single bundle $B\in\mathcal{B}_\gamma$, the optimal value $\opt{B}$ can
be computed efficiently. 

Consider any $\gamma$-bundle $B$. Contract all edges in $E[G]\setminus B$ 
to obtain an equivalent instance that corresponds to covering $B$. Since $B$
is a union of at most $\gamma$ paths, this instance has at most $2\gamma$ leaves.
Hence, to compute $\opt{B}$ it suffices to provide an algorithm that solves an 
arbitrary WTAP instance with $k$ leaves in time $n^{k^{O(1)}}$, which is the content
of the following lemma.

\begin{lemma}\label{lem:few-leaves}
 Let $(G,L,c)$ be a WTAP instance and let $U\subseteq V[G]$ be the set of leaves
in $G$. Then the WTAP instance can be solved to optimality in time 
$n^{k^{O(1)}}$, where $k = |U|$.
\end{lemma}

\begin{proof}
Let $W \subseteq V[G]$ be the nodes in $G$ with degree at least three.
 Let $Q_1,\cdots, Q_p$ be the set of paths obtained by splitting the tree
at nodes in $W$. Since there are $k$ leaves, the number
of nodes in $W$ is less than $k$, and hence $p = O(k)$. 

Let $S^* \subseteq L$ be any optimal solution. 
Now, for each pair of paths $Q_i, Q_j$, we claim that $S^*$ contains 
at most two links that connect nodes on $Q_i$ to nodes on $Q_j$. Indeed,
if this is not the case, and there are at least three such links, 
then necessarily the path of one link is contained in the union 
of the paths of the other two links, making $S^*$ redundant and contradicting
optimality (we assume that $L$ has no links of cost zero, as these links
can be included up front in any solution, and the covered edges can be contracted).

Consequently, $S^*$ has at most $O(k^2)$ links connecting nodes on different 
paths, and all other links in $S^*$ connect nodes that belong to the same path.

The algorithm starts by guessing the $O(k^2)$ links in $S^*$ that connect 
nodes on different paths. For the correct guess, the problem decomposes 
into a union of $p$ interval covering problems, which can be solved in polynomial
time using dynamic programming. The union of the optimal guess and the
optimal solutions from the $p$ paths comprises the reported optimal solution.
It is easy to verify the running time. This proves the lemma.
\end{proof}

\section{Omitted Proofs}\label{apx:proofs}

\subsection{Proof of Lemma~\ref{lem:up_rounding}}
The lemma follows from the fact that the constraint matrix of the natural 
LP is totally unimodular if all links are up-links, as this implies that 
there exists an integral solution with the same objective function value,
that can be computed by solving the natural LP restricted to the variables
corresponding only to up-links.

To prove this claim
we use the Ghouila-Houri condition on the rows of the constraint matrix. 
Consider any subset of constraints, corresponding to a subset $F\subseteq E[G]$ of edges.
We show that there exists a subset $F^+ \subseteq F$ such that the constraint
obtained by adding all constraints for edges in $F^+$ and subtracting from the
result the sum of the constraints in $F\setminus F^+$ results in a constraint
with coefficients in $\{0,1,-1\}$. This will imply, by the Ghouila-Houri
condition, that the constraint matrix is totally unimodular.

First observe that we can assume that $F=E[G]$, as the constraint matrix
that corresponds to any subset $F\subseteq E[G]$ is simply the constraint 
matrix of the natural LP for the tree obtained from $G$ by contracting
all edges in $E[G]\setminus F$. To this end define $F^+$ to be the set 
of edges that are at an odd distance from $r$, i.e.\ the set of edges $e$,
such that the number of edges different from $e$ on the path connecting 
$r$ to the closest node in $e$ is odd. 

It remains to show that for every link $\ell=uv\in L$, it holds that 
$$
|\{e\in F^+ \,\mid\, \ell\in \cov{e}\}| - |\{e\in F\setminus F^+ \,\mid\, \ell\in \cov{e}\}| \in \{0,1,-1\}.
$$
This, however, holds because $\ell$ is an up-link, so the edges in
$P_{\ell}$ alternate between belonging to $F^+$ and to $F\setminus F^+$.

\subsection{Proof of Lemma~\ref{lem:star-shaped}}

The ``only if'' direction is trivial. To prove the ``if'' direction
consider a solution $S\subseteq L$ that covers all leaf edges. We prove
that it is feasible for the WTAP instance. Consider any edge $e\in E[G]$. We
show that it is covered by $S$. Fix any hub $r\in V[G]$.
There exists some leaf $u\in V[G]$ of the tree such that $e$ lies on 
the $u$-$r$ path in $G$. Consider the link $\ell\in L$ that covers the 
leaf edge of $u$. Since the instance is star-shaped, the path $P_\ell$
is incident to $r$, and hence $e\in P_\ell$, as $P_\ell$ contains all
edges on the $u$-$r$ path in $G$. It follows that $e$ is covered.

\subsection{Proof of Lemma~\ref{lem:weight-presenving}}

The proof follows easily from the fact that in any splitting operation performed
in the algorithm, whenever the fractional assignment of some link $\ell$ covering 
some edges in $F$ is decreased to zero, the fractional assignment of
a shadow $\ell'$ of $\ell$ with the same cost is increases by the same fraction.

\subsection{Proof of Lemma~\ref{lem:weight-increase-decomp}}

Observe that the total number of splittings in the greedy procedure is 
equal to the number of parts in the final decomposition minus one, namely
$k-1$. Due to the thin coverage property, we know that in each splitting the
cost is increased by a total of at most $\frac{2M}{\epsilon}$, so the
total increase satisfies 
$$
\sum_{i\in [k]} c^\intercal z^i - c^\intercal x \leq \frac{2kM}{\epsilon}.
$$ 
Now, since we only perform splitting at $\alpha(M,\epsilon)$-thin edges, 
each pair $(T^j,z^j)$ in the decomposition satisfies 
$c^\intercal z^j \geq \frac{4M}{\epsilon^2}$, hence
$$
\sum_{i\in [k]} c^\intercal z^i \geq \frac{4kM}{\epsilon^2}, 
$$
which implies
$\frac{\sum_{i\in [k]} c^\intercal z^i - c^\intercal x}{c^\intercal x} \leq \epsilon,$ as desired.
% \frac{\frac{2kM}{\epsilon}}{\frac{4kM}{\epsilon^2} - \frac{2kM}{\epsilon}} \leq
% ...
% $$

\subsection{Proof of Lemma~\ref{lem:simple-trees}}
Recall that $x(\cov{e}) \leq \frac{2}{\epsilon}$ for all edges
$e\in E[\bar G]$. Also, the splitting operation does not increase the 
coverage level of any edge. More precisely, for any pair $(T,z)$ and any
thin edge $e = \{u,v\}\in E[T]$ that are used for splitting in the decomposition
algorithm, it holds that $z(\cov{e'}) = z^u(\cov{e'})$ and 
$z(\cov{e''}) = z^v(\cov{e''})$ for any $e'\in E[T^u]$ and $e''\in E[T^v]$.

Next, recall that pairs $(T,z)$ that comprise the final decomposition have 
no $\frac{4M}{\epsilon^2}$-thin edges. This means that for every edge $e \in E[T]$, there exists 
$u\in e$ such that
\begin{equation}\label{eq:1}
 \sum_{\ell\in L, \,\, \ell \in V[T^u]\times V[T^u]} c_\ell z_\ell < \frac{4M}{\epsilon^2}.
\end{equation}
If there exists an edge $e = \{u,v\}$ for which condition $(\ref{eq:1})$ holds for both 
endpoints $u$ and $v$, then, by the thin coverage property of the solution we have 
$$
c^\intercal z < 2 \cdot \frac{4M}{\epsilon^2} + \sum_{\ell\in \cov{e}} c_\ell z_\ell <
\frac{8M}{\epsilon^2} + \frac{2M}{\epsilon} \leq \frac{10M}{\epsilon^2},
$$
showing that any node $u\in V[T]$ satisfies the first part of the $\beta'$-simple property 
for $\beta' = \frac{10M}{\epsilon^2}$ in this case. In the other
case, every edge $e = \{u,v\}$ has \emph{exactly} one endpoint satisfying $(\ref{eq:1})$.
Direct each edge $\{u,v\}$ from the node corresponding to the subtree satisfying
$(\ref{eq:1})$ to the node corresponding to the subtree not satisfying $(\ref{eq:1})$.
% We claim that the resulting directed tree has no node with out-degree
% larger than one. Indeed assume that there exist $u\in V[T]$ and 
% $e = \{u,v\}, e' = \{u,w\} \in E[T]$ such that both $e$ and $e'$ are
% directed from $u$. It follows that the subtree on the $v$-side of $e$
% and the subtree on the $w$-side of $e'$ both \emph{do not} satisfy $(\ref{eq:1})$. 
% Since $e$ separates these subtrees, it means that $e$ has 
% links with cost at least $\frac{4M}{\epsilon^2}$ with respect to $z$ on each side,
% making it a thin edge, and contradicting our assumption that $(T,z)$ is 
% part of the final decomposition (Recall that the decomposition algorithm
% reports a pair as part of the final decomposition only if it contains no
% thin edge).
% Now, since the obtained directed tree has no nodes with out-degree larger
% than one, 
Now, let $u\in V[T]$ be any node with out-degree zero. By definition
of the directed tree, the removal of this node results in 
subtrees $K_1, \cdots, K_p$ satisfying
$$
\sum_{\ell \in L, \,\, \ell \in V[K_j]\times V[K_j]} c_\ell z_\ell \leq \frac{4M}{\epsilon^2}
% \leq \frac{8M^2}{\epsilon^2}
$$
for all $j\in [p]$, proving existence of a node certifying the first part 
of the $\beta'$-simple property for this case as well. 

Until now we proved the existence of a node $u\in V[T]$, the removal of which
results in a set of trees $K_1, \cdots, K_p$ satisfying the first part 
of the $\beta'$-simple property, for $\beta' = \frac{10M}{\epsilon^2}$. 
We claim that any such a node is a $\frac{48M}{\epsilon^2}$-center.
Since $\frac{48M}{\epsilon^2} > \frac{10M}{\epsilon^2}$, the first part
of the property clearly holds, so it remains to prove that each subtree
$K_j$ has at most $\frac{48M}{\epsilon^2}$ leaves. Assume towards contradiction 
that for some $j\in [p]$, the tree $K_j$ has more than $\frac{48M}{\epsilon^2}$ 
leaves. Since each link has cost at least $1$, and one link can cover at most
two leaf edges, the optimal solution to the WTAP instance restricted to the
subtree $K_j$ has cost greater than 
$\frac{1}{2} \cdot \frac{48M}{\epsilon^2} = \frac{24M}{\epsilon^2}$.
Consequently, due to Proposition~\ref{prop:simple_rounding}, the cost of any 
fractional solution to the natural LP relaxation on this instance has cost
greater than $\frac{1}{2} \cdot \frac{24M}{\epsilon^2} = \frac{12M}{\epsilon^2}$.
This leads to a contradicting as follows. Let $e'$ be the edge connecting $u$
to the subtree $K_j$. Then, $z$ fractionally covers $K_j$ using only the 
links in $L' = V[K_j]\times V[K_j] \cup \cov{e'}$, while 
$$
\sum_{\ell \in L'} c_\ell z_\ell \leq \frac{10M}{\epsilon^2} + \frac{2M}{\epsilon} \leq \frac{12M}{\epsilon^2},
$$
where we used $\sum_{\ell \in L, \,\, \ell \in V[K_j]\times V[K_j]} c_\ell z_\ell \leq \frac{10M}{\epsilon^2}$
and $\sum_{\ell\in \cov{e'}} c_\ell z_\ell \leq \frac{2M}{\epsilon}$.
This concludes the proof of the lemma.
% $$
% \sum_{\ell \in L, \,\, \ell \in V[K_j]\times V[K_j]} c_\ell z_\ell + 
% \sum_{} c_\ell z_\ell \leq \frac{10M}{\epsilon^2} + \frac{2M}{\epsilon} \leq .
% $$

\section{Better Rounding for TAP (Proof of Theorem~\ref{thm:tap})}\label{apx:tap}

To prove Theorem~\ref{thm:tap} we only need to prove the following 
version of Lemma~\ref{lem:cross-rounding}. Recall that $(T,z)$
is a pair in the instance decomposition, $T$ is rooted at $r\in V[T]$, which is
a $\beta$-center of $T$, and that $z^{in}$ and $z^{cr}$ are the parts of $z$ 
corresponding to in-links and cross-links, respectively.

\begin{lemma}\label{lem:cross-rounding-tap}
Given $T$ and $z$, there is a polynomial time algorithm that produces a set 
of links $S\subseteq L$ that covers $E[T]$ with cost at most 
$$
|S| \leq 2 z^{in}(L) + \frac{3}{2} z^{cr}(L).
$$
\end{lemma}

\begin{proof}
 As in the proof of Lemma~\ref{lem:cross-rounding}, we start by replacing
each in-link in the support of $z$ with its two shadows. This results in 
a fractional solution $y \in \mathbb{R}_{\geq 0}^L$ with only cross-links
and in-links that are also up-links in the support. We also have 
$$
y(L) \leq 2z^{in}(L) + z^{cr}(L),
$$
as before. 
% Denote by $U\subseteq V[T]$ the set of leaves of $T$. 
% We say that a cross-link $\ell$ is \emph{ready} if it connects two leaves of $T$. 
The rounding algorithm proceeds as follows. The current tree and the current 
fractional solutions are denoted by $T'$ and $y'$, respectively.

\begin{enumerate}[I]
 \item If some leaf edge of the current tree is covered only by up-links: Include
the up-link $\ell$ covering this leaf edge that covers the most edges, i.e.\
for which $P^{T'}_\ell$ is the largest. 
% Reduce to zero the fraction in $y'$ of all shadows of $\ell$. 
Contract the link $\ell$.
Go to $\mathrm{I}$.
 \item Else, if there exists a link $\ell$ connecting two leaves $u,v \in V[T']$,
choose an arbitrary such link, include it in the solution, and contract all covered 
edges. Go to $\mathrm{I}$.
 \item Else, choose \emph{for every} leaf of $T'$, one cross-link covering it, and
include it in the solution. Return the obtained solution.
\end{enumerate}

We assume that after each contraction operation in the algorithm, all links
that become self-loops are removed from the support of $y'$.
We claim that the latter procedure always terminates with a feasible solution $S$
with at most 
$$
y^{in}(L) + \frac{3}{2} y^{cr}(L) = 2z^{in}(L) + \frac{3}{2} z^{cr}(L)
$$
links. We make a few observations. First, notice that links included in 
step $\mathrm{I}$ of the solution do not incur loss in terms of the fractional
solution. Indeed, in these iterations a single link is added to the solution,
while the total fractional value of the current solution also drops by one unit,
since the fractional solution is feasible, and hence the covered leaf edge has
at least one unit of fractional links incident to it. 

Next, we claim that also in step $\mathrm{III}$ no loss is incurred with respect
to the fractional solution. Indeed, step $\mathrm{III}$ can only be reached 
if each leaf edge is covered only by fractional up-links and cross-link 
that connect the corresponding leaf with an internal node in $T'$. Furthermore,
since step $\mathrm{I}$ did not materialize, each leaf must have at least one 
incident such cross-link that is in the support of $y'$. 
It follows that, on the one hand, no link in the support
of $y'$ covers more than one leaf edge, and on the other hand, it is possible
to choose one cross-link incident to each leaf edge. Let $q$ denote the number of leaf edges. 
The first property implies that $y'(L) \geq q$, as all leaf edges are 
fractionally covered by $y'$.
The second property implies that there is a set of $q$ cross-links,
that cover all leaf edges in $T'$. Since the instance restricted to cross links 
is star-shaped (with the root $r$ as its hub), it suffices to choose these $q$ links to cover all of $T'$.
This also proves that the algorithm returns a feasible solution.

It remain to analyze the loss incurred by step $\mathrm{II}$. 
Since steps $\mathrm{I}$ and $\mathrm{III}$ incur no loss in terms of the
fractional cost, the size of the set $S$ returned by the algorithm is
at most 
$$
y(L) + \sum_{\ell\in \bar L} (1-y_\ell) = y(L) + \sum_{\ell\in \bar L} (1-z_\ell), 
$$ 
where $\bar L$ is the set of cross-links contracted in 
step $\mathrm{II}$ of the algorithm. Denote
$\Delta = \sum_{\ell\in \bar L} (1-z_\ell)$. We show that 
$\Delta \leq \frac{1}{2}z^{cr}(L)$. This clearly suffices to prove
the lemma.

Let $\ell = uv \in \bar L$ be a link contracted in step $\mathrm{II}$, 
where $u$ and $v$ are leaves of $T'$, and let $e_u$ and $e_v$ be 
their respective leaf edges. Let $y'$ be the fractional solution in the beginning
of that iteration.
Let $L_u$ and $L_v$ be the links different from $\ell$ in the support of $y'$ that cover $e_u$ and
$e_v$, respectively. From feasibility of $y'$ it follows that
$$
\sum_{\ell' \in L_p} y'_{\ell'} \geq 1-y'_\ell
$$
holds for $p = u,v$. Furthermore, since $L_u \cap L_v = \emptyset$, it holds that  
$\sum_{\ell' \in L_u\cup L_v} y'_{\ell'} \geq 2(1-y'_\ell).$ Define 
$R = L_u\cup L_v$ and $\theta = 1-y'_\ell$. We can assume without loss
of generality that $R$ contains no in-links. Indeed, every in-link
in $R$ is a shadow of $\ell$, and hence before contracting $\ell$, 
we can shift weight from any such shadow to $\ell$, without increasing
the cost of the solution, while maintaining feasibility.

Now, notice that by contracting $\ell$ the links $R$ become up-links, 
and are hence never used as contracted links in step $\mathrm{II}$ 
in later iterations. 
Let $\bar y \in \mathbb{R}_{\geq 0}^L$ be the fractional solution in 
the end of the current iteration. From the latter considerations we have 
that 
$$
\bar y^{cr}(L) \leq {y'}^{cr}(L) - \sum_{\ell'\in R} y'_{\ell'} \leq {y'}^{cr}(L) - 2\theta.
$$
Now, since throughout the algorithm we never increase the 
fraction on any cross-links that we do not immediately contract, it follows
immediately from the previous inequality that $2\Delta \leq y^{cr}(L) = z^{cr}(L)$,
implying $\Delta \leq \frac{1}{2}z^{cr}(L)$ and proving the lemma.

\end{proof}

We can now prove Theorem~\ref{thm:tap} by combining Lemmas~\ref{lem:bundle-rounding}
and~\ref{lem:cross-rounding-tap}.

\begin{proof}[Proof of Theorem~\ref{thm:tap}]
 The proof is identical to that of Theorem~\ref{thm:wtap}, except that we use
Lemma~\ref{lem:cross-rounding-tap} instead of Lemma~\ref{lem:cross-rounding}, and
we use a different threshold based on the ratio $\frac{z^{cr}(L)}{z(L)}$ (instead of
$\alpha^*$). Concretely, we use Lemma~\ref{lem:cross-rounding-tap} to round the
pair $(T,z)$ whenever $\frac{z^{cr}(L)}{z(L)}\geq \frac{2}{3}$, and otherwise we use
Lemma~\ref{lem:bundle-rounding}. The obtained approximation guarantee is not 
worse than $\frac{5}{3} +\epsilon$ in both cases.
\end{proof}

\end{document}